\documentclass[a4paper,UKenglish]{lipics-v2018}

\nolinenumbers

\usepackage{algorithm}
\usepackage[noend]{algpseudocode}

\usepackage{wrapfig}

\newcommand{\D}{\mathbb{D}}

\renewcommand{\S}{\mathbb{S}}

\renewcommand{\L}{\ensuremath{\mathcal{L}}}

\newcommand{\dupcost}{\hat{d}}
\newcommand{\lcamap}{\mu}

\newcommand{\map}{\alpha}

\newcommand{\G}{\mathcal{G}}

\renewcommand{\P}{\mathcal{P}}

\newtheorem{nclaim}{Claim}

\title{Reconciling Multiple Genes Trees via Segmental Duplications and Losses}
\titlerunning{Reconciling Multiple Genes with Segmental Duplications}  
\author{Riccardo Dondi}{Dipartimento di Filosofia, Lettere, Comunicazione, Universit\`a degli Studi di Bergamo, Bergamo, Italy,}{riccardo.dondi@unibg.it}{}{}
\author{Manuel Lafond}{Department of Computer Science, Universit\'e de Sherbrooke, Qu\'ebec, Canada,}{manuel.lafond@USherbrooke.ca}{}{}
\author{Celine Scornavacca}{ISEM, CNRS, Universit\'e de Montpellier, IRD, EPHE, Montpellier, France,}{celine.scornavacca@umontpellier.fr}{}{}

\authorrunning{R. Dondi, M. Lafond and C. Scornavacca} 

\Copyright{Riccardo Dondi, Manuel Lafond and Celine Scornavacca}


\category{}

\relatedversion{}

\supplement{}

\funding{}

\acknowledgements{The authors would like to thank Mukul Bansal for providing the eukaryotes data set for the experiments section.}

\EventEditors{Laxmi Parida and Esko Ukkonen}
\EventNoEds{2}
\EventLongTitle{18th International Workshop on Algorithms in Bioinformatics (WABI 2018)}
\EventShortTitle{WABI 2018}
\EventAcronym{WABI}
\EventYear{2018}
\EventDate{August 20--22, 2018}
\EventLocation{Helsinki, Finland}
\EventLogo{}
\SeriesVolume{113}
\ArticleNo{5} 

\begin{document}

\maketitle              

 \begin{abstract}
 Reconciling gene trees with a species tree is a fundamental problem to understand the evolution
 of gene families. Many existing approaches reconcile each gene tree independently. However, 
 it is well-known that the evolution of gene families is interconnected. 
 In this paper, we extend a previous approach to reconcile a set of gene trees with a species tree 
 based on segmental macro-evolutionary events, 
 where segmental duplication events and losses are associated with cost $\delta$ and $\lambda$, respectively.
 We show that the problem is polynomial-time solvable when $\delta \leq \lambda$ (via LCA-mapping),
 while if $\delta > \lambda$ the problem is NP-hard, even when $\lambda = 0$ and a single
 gene tree is given, solving a long standing
 open problem on the complexity of the reconciliation problem.
 On the positive side, we give a fixed-parameter algorithm for the problem, where the parameters are $\delta/\lambda$
 and the number $d$ of segmental duplications, of time complexity 
 $O(\lceil \frac{\delta}{\lambda} \rceil^{d} \cdot n \cdot \frac{\delta}{\lambda})$. 
Finally, we demonstrate the usefulness of this algorithm on two previously studied real datasets: we first show that our method can be used to confirm or refute hypothetical segmental duplications on a set of 16 eukaryotes, then show how we can detect whole genome duplications in yeast genomes.
\end{abstract}

\section{Introduction}
It is nowadays well established that the evolution of a gene family can differ from that of the species containing these genes. This can be due to quite a number of different reasons, including gene duplication, gene loss, horizontal gene transfer or incomplete lineage sorting, to only name a few \cite{maddison1997gene}. 
While this discongruity between the gene phylogenies  (the \emph{gene trees}) and the species phylogeny (the \emph{species tree}) complicates the process of reconstructing the latter from the former
, every cloud has a silver lining: ``plunging''  gene trees into the species tree and analyzing the differences between these topologies, one can gather  hints to unveil the macro-evolutionary events that shaped gene evolution. This is the rationale behind the \emph{species tree-gene tree reconciliation}, a concept introduced in  \cite{goodman1979fitting} and first formally defined in \cite{page1994maps}. Gathering intelligence on these macro-evolutionary events permits us to better understand the mechanisms of evolution with applications ranging from orthology detection~\cite{lafond2018accurate,ullah2015integrating} to ancestral genome reconstruction~\cite{duchemin2017decostar}, and recently in dating phylogenies~\cite{davin2018gene,Chauve127548}. 

It is well-known that the evolution of gene families is interconnected. However, in current pipelines, each gene tree is reconciled independently with the species tree, even when posterior to the reconciliation phase the genes are considered as related, e.g. \cite{duchemin2017decostar}. A more pertinent approach would be to reconcile the set of gene trees at once and consider \emph{segmental} macro-evolutionary events, i.e. that concern a chromosome segment instead of a single gene.   

Some work has been done in the past to model segmental gene duplications and three models have been considered: the {\sc EC} (\emph{Episode Clustering}) problem, the {\sc ME} (\emph{Minimum Episodes}) problem \cite{guigo1996reconstruction,bansal2008multiple}, and the {\sc MGD} (\emph{Multiple Gene Duplication}) problem \cite{fellows1998multiple}.  
The {\sc EC} and {\sc MGD} problems both aim at clustering duplications together by minimizing the number of locations in the species tree where at least one duplication occurred, with the additional requirement that a cluster cannot contain two gene duplications from a same gene tree in the {\sc MGD} problem. The 
 {\sc ME} problem is  more biologically-relevant, because it aims at minimizing the actual number of segmental duplications (more details in Section \ref{def:ME}). 
Most of the exact solutions proposed for the {\sc ME} problem \cite{bansal2008multiple,luo2011linear,paszek2017efficient} deal with a constrained version,
since the possible mappings of a gene tree node are limited to given intervals, see for example  \cite[Def. 2.4]{bansal2008multiple}. 
In~\cite{paszek2017efficient}, a simple $O^*(2^{k})$ time algorithm is presented for the unconstrained version (here $O^*$ hides polynomial factors), where $k$ is the number of speciation nodes that have descending duplications under the LCA-mapping
This shows that the problem is fixed-parameter tractable (FPT) in $k$.  But since the LCA-mapping maximizes the number of speciation nodes, there is no reason to believe that $k$ is a small parameter, and so more practical FPT algorithms are needed.

In this paper, we extend the unconstrained ME model to gene losses and provide a variety of new algorithmic results.
We allow weighing segmental duplication events and loss events by separate costs $\delta$ and $\lambda$, respectively.
We show that if $\delta \leq \lambda$, then an optimal reconciliation can be obtained by reconciling each gene tree separately under the usual LCA-mapping, even in the context of segmental duplications.
On the other hand, we show that if $\delta > \lambda$ and both costs are given, 
reconciling a set of gene trees while minimizing segmental gene duplications and  gene losses is NP-hard.  The hardness also holds in the particular case that we ignore losses, i.e. when $\lambda = 0$. This solves a long standing open question on the complexity of the reconciliation problem under this model (in \cite{bansal2008multiple}, the authors already said 
\emph{``it would be interesting to extend the [...] model of Guigó et al. (1996) by relaxing the constraints
on the possible locations of gene duplications on the species
tree''}. The question is stated as still open in \cite{paszek2017efficient}). The hardness holds   
also when only a single gene tree is given. On the positive side, we describe an algorithm that is practical when $\delta$ and $\lambda$ are not too far apart.  More precisely, we show that 
multi-gene tree reconciliation is fixed-parameter tractable in the ratio $\delta/\lambda$ and the number $d$ of segmental duplications, and can be solved in time $O(\lceil \frac{\delta}{\lambda} \rceil^{d} \cdot n \cdot \frac{\delta}{\lambda})$.  
The algorithm has been implemented and tested and is available\footnote{To our knowledge, this is the first publicly available reconciliation software for segmental duplications.} at \url{https://github.com/manuellafond/Multrec}.  We first evaluate the potential of multi-gene reconciliation on a set of 16 eukaryotes, and show that our method can find scenarios with less duplications than other approaches.  While some previously identified segmental duplications are confirmed by our results, it casts some doubt on others as they do not occur in our optimal scenarios.  We then show how the algorithm can be used to detect whole genome duplications in yeast genomes.
Further work includes incorporating in the model segmental gene losses and segmental horizontal gene transfers, with a similar flavor than the heuristic method discussed in \cite{phdDuchemin}.
\section{Preliminaries}

\subsection{\bf Basic notions}

For our purposes, a {\em rooted phylogenetic tree} $T=(V(T),E(T))$ is an oriented tree, where $V(T)$ is the set of nodes, $E(T)$ is the set of arcs, 
all oriented away from $r(T)$, the root. Unless stated otherwise, all trees in this paper are 
rooted phylogenetic trees.
A \emph{forest} $F = (V(F), E(F))$ is a directed graph in which every connected component is a tree.  Denote by $t(F)$ the set of trees of $F$ that are formed by its connected components.
Note that a tree is itself a forest.  In what follows, we shall extend the usual terminology on trees to forests.

For an arc~$(x,y)$ of $F$, we call $x$ the \emph{parent} of $y$, and $y$ a \emph{child} of $x$.
If there exists a path that starts at $x$ and ends at $y$, then $x$ is an \emph{ancestor} of $y$ and $y$ is a \emph{descendant} of $x$.  We say $y$ is a \emph{proper} descendant of $x$ if $y \neq x$,
and then $x$ is a proper ancestor of $y$.
This defines a partial order denoted by $y \leq_{F} x$, and $y <_F x$ if $x\not=y$ 
(we may omit the $F$ subscript if clear from the context).
If none of $x \leq y$ and $y \leq x$ holds, then $x$ and $y$ are \emph{incomparable}.
The set of children of $x$ is denoted $ch(x)$ and its parent $x$ is denoted $par(x)$ 
(which is defined to be $x$ if $x$ itself is a root of a tree in $t(F)$).  For some integer $k \geq 0$, 
we define $par^k(x)$ as the $k$-th parent of $x$.
Formally, $par^0(x) = par(x)$ and $par^k(x) = par(par^{k - 1}(x))$ for $k > 0$.
The number of children $|ch(x)|$ of $x$ is called the \emph{out-degree} of $x$.
Nodes with no children are {\em leaves}, all others are {\em internal nodes}.
The set of leaves of a tree $F$ is denoted by $L(F)$.
The leaves of $F$ are bijectively labeled by a set $\L(F)$ of labels.
A forest is {\em binary} if $|ch(x)|=2$ for all internal nodes~$x$. 
Given a set of nodes $X$ that belong to the same tree $T \in t(F)$, the \emph{lowest common ancestor} of $X$, denoted $LCA_F(X)$, is the node $z$ that satisfies $x \leq z$ for all $x \in X$ and such that no child of $z$ satisfies this property.  We leave $LCA_F(X)$ undefined if no such node exists (when elements of $X$ belong to different trees of $t(F)$).  We may write $LCA_F(x, y)$ instead of $LCA_F(\{x,y\})$.
The \emph{height} of a forest $F$, denoted $h(F)$, is the number of nodes of a longest directed path from a root to a leaf in a tree of $F$ (note that the height is sometimes defined as the number of arcs on such a path - here we use the number of nodes instead).
Observe that since a tree is a forest, all the above notions also apply on trees.

\subsection{\bf Reconciliations}
\label{sec:rec}

A reconciliation usually involves two rooted phylogenetic trees, a {\em gene tree}~$G$ and a {\em species tree}~$S$, which we always assume to be both binary.  In what follows, we will instead define reconciliation 
between a gene forest $\G$ and a species tree.  Here $\G$ can be thought of as a set of gene trees.
Each leaf of $\G$ represents a distinct extant gene, and $\G$ and $S$ are related by a function $s: \L(\G) \rightarrow \L(S)$, which means that each extant gene belongs to an extant species.
Note that $s$ does not have to be injective (in particular, several genes from a same gene tree $G$ of $\G$ 	
can belong to the same species) or surjective (some species may not contain any gene of $\G$). Given $\G$ and $S$, we will implicitly assume the existence of the function $s$. 

In a $\mathbb{DL}$ reconciliation, each node of $\G$ is associated to a node of $S$ and an event --
a speciation ($\S$), a duplication ($\D$) or a contemporary event  ($\mathbb{C}$) -- under some constraints. 
A contemporary event $\mathbb{C}$ associates a leaf $u$ of $\G$ with a leaf $x$ of $S$ such that $s(u) = x$. A speciation in a node $u$ of $\G$ is constrained to the existence of two separated paths from the mapping of $u$ to the mappings of its two children, while the only constraint given by a duplication event  is that evolution of $\G$ cannot go back in time. 
More formally:

\begin{definition}[Reconciliation]\label{def:rec}
Given a gene forest $\G$ and a species tree $S$, a \emph{reconciliation} between $\G$ and $S$ is a function $\alpha$ that maps each node $u$ of $\G$ to a pair $(\alpha_r(u),\alpha_e(u))$ where $\alpha_r(u)$ is a node of $V(S)$ and $\alpha_e(u)$ is  an event of type $\S, \D$ or $\mathbb{C}$, such that:

\begin{enumerate}
\item
if $u$ is a leaf of $\G$, then $\alpha_e(u) = \mathbb{C}$ and $\alpha_r(u) = s(u)$;
\item
if $u$ is an internal node of $\G$ with children $u_1, u_2$, then exactly one of following cases holds:
\begin{itemize}
\item[$\bullet$]
$\alpha_e(u) = \mathbb{S}$, $\alpha_r(u) = LCA_S(\alpha_r(u_1),\alpha_r(u_2))$ and 
$\alpha_r(u_1),\alpha_r(u_2)$ are incomparable;
\item[$\bullet$]
$\alpha_e(u) = \mathbb{D}$, $\alpha_r(u_1) \leq \alpha_r(u)$ and $\alpha_r(u_2) \leq \alpha_r(u)$ 
\end{itemize}
\end{enumerate}
\label{def:recn} 
\end{definition}

Note that if $\G$ consists of one tree, this definition coincides with the usual one given in the literature (first
formally defined in \cite{page1994maps}). 
We say that $\alpha$ is an \emph{LCA-mapping} if, for each internal node $u \in V(\G)$ with children $u_1, u_2$, 
$\alpha_r(u) = LCA_S( \alpha_r(u_1), \alpha_r(u_2) )$.  Note that there may be more than one LCA-mapping, 
since the $\mathbb{S}$ and $\mathbb{D}$ events on internal nodes can vary.
The number of duplications of $\alpha$, denoted by $d(\alpha)$ is the number of nodes $u$ of $\G$ such that $\alpha_e(u)=\D$. For counting the losses, first define for $y \leq x$ the distance $dist(x, y)$ as the number of arcs on the path from $x$ to $y$.
Then, for every internal node $u$ with children  $\{u_1,u_2\}$, the number of losses associated with $u$ in a reconciliation $\alpha$, denoted by $l_{\alpha}(u)$, is defined as follows:

 \begin{itemize}
\item[$\bullet$]  
if $\alpha_e(u)=\S$, then $l_{\alpha}(u) = dist(\alpha_r(u), \alpha_r(u_1)) + dist(\alpha_r(u), \alpha_r(u_2)) - 2$; 

\item[$\bullet$] 
if $\alpha_e(u)=\D$, then $l_{\alpha}(u) = dist(\alpha_r(u), \alpha_r(u_1)) + dist(\alpha_r(u), \alpha_r(u_2))$.
\end{itemize}

The number of losses of a reconciliation $\alpha$, denoted by $l(\alpha)$, is the sum of $l_{\alpha}(\cdot)$ for all internal nodes of $\G$. The usual cost of $\alpha$, denoted by $cost(\alpha)$, is $d(\alpha) \cdot \delta+l(\alpha) \cdot \lambda$ \cite{ma2000gene}, where $\delta$ and $\lambda$ are respectively the cost of a duplication and a loss event (it is usually assumed that speciations do not incur cost).  
A \emph{most parsimonious reconciliation}, or MPR, is a reconciliation $\alpha$ of minimum cost.
It is not hard to see that finding such an $\alpha$ can be achieved by computing a MPR
for each tree in $t(\G)$ separately.  This MPR is the unique LCA-mapping $\alpha$
in which $\alpha_e(u) = \mathbb{S}$ whenever it is allowed according to Definition~\ref{def:rec} \cite{chauve2009new}.

\subsection{Reconciliation with segmental duplications\label{def:ME}}

Given a reconciliation $\alpha$ for $\G$ in $S$, and given $s \in V(S)$, 
write $D(\G, \alpha, s) = \{u \in V(\G) : \alpha_e(u) = \mathbb{D}$ and $\alpha_r(u) = s\}$ for the set of duplications of $\G$ mapped to $s$.
We define $\G[\alpha, s]$ to be the subgraph of $\G$ induced by the nodes in $D(\G, \alpha, s)$.
Note that $\G[\alpha, s]$ is a forest.

\begin{figure}[H]
\begin{center}
\includegraphics[width=0.95\linewidth]{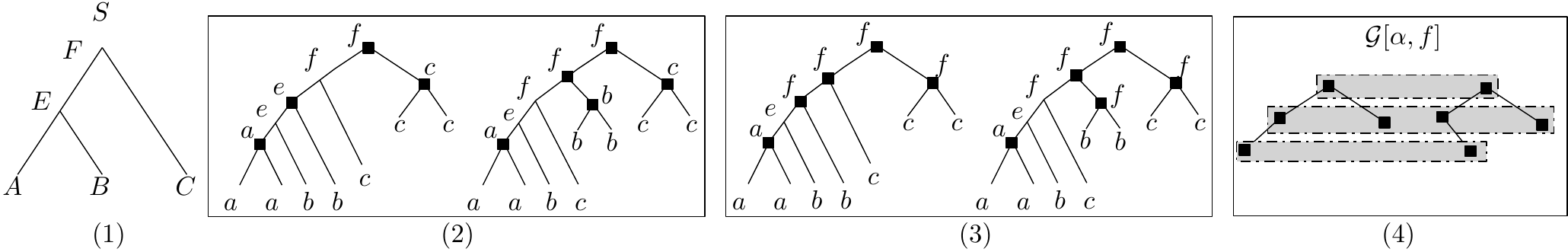}
\caption{(1) A species tree $S$.  (2) A gene forest $\G$ with two gene trees reconciled under the MPR that we denote $\lcamap$.  The nodes are labeled by the lowercase name of the species they are mapped to.  Black squares indicate duplication nodes.  Losses are not shown.  (3) The same forest $\G$ but with another reconciliation $\map$ for the internal nodes.  (4) The forest $\G[\map, f]$, along with a partition into (possible) segmental duplications.
\label{fig:example}}
\end{center}
\end{figure}

Here we want to tackle the problem of reconciling several gene trees at the same time and counting segmental duplications only once. Given a set of duplications nodes $\mathcal{D} \in V(\G)$ occurring in a given node $s$ of the species tree, it is easy to see that the minimum the number of segmental 
duplications associated with $s$ is the minimal number of parts in a partition of $\mathcal{D}$ in which each part does not contain comparable nodes.  See Figure~\ref{fig:example}.(4) for an example.  This number coincides \cite{bansal2008multiple}  with $h_{\alpha}(s) := h(\G[\alpha, s])$, i.e. the height of the forest of the duplications in $s$. 
Now, denote $\dupcost(\alpha) = \sum_{s \in V(S)} h_{\alpha}(s)$.
For instance in Figure~\ref{fig:example}, under the mapping $\lcamap$ in (2), we have $\dupcost(\lcamap) = 6$, because $h_{\lcamap}(s) = 1$ for $s \in \{A,B,C,E\}$ and $h_{\lcamap}(F) = 2$.  
But under the mapping $\map$ in (3), $\dupcost(\map) = 4$, since 
$h_{\map}(A) = 1$ and $h_{\map}(F) = 3$.  Note that this does not 
consider losses though --- the $\map$ mapping has more losses than $\lcamap$.

The cost of $\alpha$ is $cost^{SD}(\G, S, \alpha) = \delta \cdot \dupcost(\alpha) + \lambda \cdot l(\alpha)$.
If $\G$ and $S$ are unambiguous, we may write $cost^{SD}(\alpha)$.
We have the following problem :

\vspace{3mm}

\noindent
{\sc Most Parsimonious Reconciliation of a Set of Trees with Segmental Duplications ({\sc MPRST-SD})}\\
{\bf Instance:} a species tree $S$, a gene forest $\G$, costs $\delta$ for duplications and $\lambda$ for losses.\\
{\bf Output:} a reconciliation $\alpha$ of $\G$ in $S$ such that  
$cost^{SD}(\G, S, \alpha)$ is minimum.

\vspace{3mm}

Note that, when $\lambda=0$, $cost^{SD}$ coincides with the unconstrained ME score defined in \cite{paszek2017efficient} (where it is called the FHS model).



\subsection{Properties of multi-gene reconciliations}

We finish this section with some additional terminology and general properties of multi-gene reconciliations that will be useful throughout the paper.
The next basic result states that in a reconciliation $\alpha$, 
we should set the events of internal nodes to $\mathbb{S}$ whenever 
it is allowed.

\begin{lemma}\label{lem:no-useless-dups}
Let $\alpha$ be a reconciliation for $\G$ in $S$, and let $u \in V(\G)$ 
such that $\alpha_e(u) = \mathbb{D}$.
Let $\alpha'$ be identical to $\alpha$, with the exception that 
$\alpha'_e(u) = \mathbb{S}$, and suppose that $\alpha'$ satisfies the requirements of Definition~\ref{def:rec}.
Then $cost^{SD}(\alpha') \leq cost^{SD}(\alpha)$.
\end{lemma}

\begin{proof}
Observe that changing $\alpha_e(u)$ from $\mathbb{D}$ to $\mathbb{S}$ cannot increase $\dupcost(\alpha)$.
Moreover, as $dist(\alpha'_r(u), \alpha'_r(u_1))$ and $dist(\alpha'_r(u), \alpha'_r(u_2))$ are the same as in $\alpha$ for the two children $u_1$ and $u_2$ of $u$, by definition of duplications and losses
this decreases the number of losses by $2$.  Thus $cost^{SD}(\alpha') \leq cost^{SD}(\alpha)$, and this inequality is strict when $\lambda > 0$.
\end{proof}

Since we are looking for a most parsimonious reconciliation, by Lemma~\ref{lem:no-useless-dups} we may 
assume that for an internal node $u \in V(\G)$, 
$\alpha_e(u) = \mathbb{S}$ whenever allowed, and $\alpha_e(u) = \mathbb{D}$ otherwise.
Therefore, $\alpha_e(u)$ is implicitly defined by the $\alpha_r$ mapping.
To alleviate notation, we will treat $\alpha$ as simply as a mapping from $V(\G)$ to $V(S)$ and thus write $\alpha(u)$ instead of $\alpha_r(u)$.  We will assume that the events $\alpha_e(u)$ can be 
deduced from this mapping $\alpha$ by Lemma~\ref{lem:no-useless-dups}.

Therefore, treating $\map$ as a mapping, 
we will say that $\map$ is \emph{valid} if for every $v \in V(\G)$, $\map(v) \geq \map(v')$ for all descendants $v'$ of $v$.
We denote by $\map[v \rightarrow s]$ the mapping obtained from $\map$ by remapping $v \in V(\G)$ to $s \in V(S)$, i.e.
$\map[v \rightarrow s](w) = \map(w)$ for every $w \in V(\G) \setminus \{v\}$, and $\map[v \rightarrow s](v) = s$.
Since we are assuming that $\mathbb{S}$ and $\mathbb{D}$ events can be deduced from $\alpha$, 
the LCA-mapping is now unique: we denote by $\lcamap : V(\G) \rightarrow V(S)$ the LCA-mapping, defined as 
$\lcamap(v) = s(v)$ if $v \in L(\G)$, and otherwise $\lcamap(v) = LCA_{S}( \lcamap(v_1), \lcamap(v_2) )$, 
where $v_1$ and $v_2$ are the children of $v$.  
Note that for any valid reconciliation $\alpha$, we have $\map(v) \geq \lcamap(v)$ for all $v \in V(\G)$.  We also have the following, which will be useful to establish our results.

\begin{lemma}\label{lem:remap-dups}
Let $\map$ be a mapping from $\G$ to $S$.  If $\map(v) > \lcamap(v)$, then $v$ is a $\D$ node under $\map$.
\end{lemma}

\begin{proof}
Let $v_1$ and $v_2$ be the two children of $v$.
If $\map(v) \neq LCA_{S}(\map(v_1), \map(v_2))$, then $v$ must be a duplication, by the definition 
of $\S$ events.  The same holds if $\map(v_1)$ and $\map(v_2)$ are not incomparable. 
Thus assume $\map(v) = LCA_{S}(\map(v_1), \map(v_2)) > \lcamap(v)$
and that $\map(v_1)$ and $\map(v_2)$ are incomparable.
This implies that one of $\map(v_1)$ or $\map(v_2)$ 
is incomparable with $\lcamap(v)$, say $\map(v_1)$ w.l.o.g.. 
But $\lcamap(v_1) \leq \map(v_1)$, implying that $\lcamap(v_1)$ 
is also incomparable with $\lcamap(v)$, a contradiction to the definition of $\lcamap = LCA_S(\lcamap(v_1), \lcamap(v_2))$.
\end{proof}

\begin{lemma}\label{lem:dups-from-below}
Let $\map$ be a mapping from $\G$ to $S$, and let $v \in V(\G)$.
Suppose that there is some proper descendant $v'$ of $v$ such that 
$\map(v') \geq \lcamap(v)$.  Then $v$ is a duplication under $\map$.
\end{lemma}

\begin{proof}
If $\map(v) = \lcamap(v)$, we get $\lcamap(v) \leq \map(v') \leq \map(v) = \lcamap(v)$, and so $\map(v') = \lcamap(v)$. 
We must then have $\map(v'') = \lcamap(v)$ for every node $v''$ on the path between $v'$ and $v$.
In particular, $v$ has a child $v_1$ with $\map(v) = \map(v_1)$ and thus $v$ is a duplication.  If instead $\map(v) > \lcamap(v)$, then $v$ is a duplication by Lemma~\ref{lem:remap-dups}.
\end{proof}

The \emph{Shift-down lemma} will prove very useful to argue that we should shift mappings of duplications down when possible, as it saves losses (see Figure~\ref{fig:shift-down}).
For future reference, do note however that this may increase the height of some duplication forest $\G[\alpha, s]$.

\begin{lemma}[Shift-down lemma]\label{lem:shift-down}
Let $\alpha$ be a mapping from $\G$ to $S$, let $v \in V(\G)$, let $s \in V(S)$ and $k > 0$ be such that $par^k(s) = \map(v)$.  Suppose that 
$\map[v \rightarrow s]$ is a valid mapping.
Then $l(\map[v \rightarrow s]) \leq l(\map) - k$.
\end{lemma}

\begin{proof}

Let $v_1$ and $v_2$ be the children of $v$, 
and denote $t := \map(v), t_1 := \map(v_1)$ and $t_2 := \map(v_2)$.
Moreover denote $\map' := \map[v \rightarrow s]$.
Let $P$ be the set of nodes that appear on the path between $s$ and $t$, excluding $s$ but including $t$
(note that $s$ is a proper descendant of $t$ but an ancestor of both $t_1$ and $t_2$, by the validity of $\map'$). For instance in Figure~\ref{fig:shift-down}, $P = \{x,t\}$.
Observe that $|P| = k$.
Under $\map$, there is a loss for each node of $P$ on both the $(v, v_1)$ and $(v, v_2)$ branches.
(noting that $v$ is a duplication by Lemma~\ref{lem:remap-dups}).
These $2k$ losses are not present under $\map'$.  On the other hand, there are at most $k$ losses
that are present under $\map'$ but not under $\map$, which consist of one loss for each node of $P$
on the $(par(v), v)$ branch (in the case that $v$ is not the root of its tree - otherwise, 
no such loss occurs).  This proves that $l(\map') \leq l(\map) - k$.
\end{proof}

An illustration of the  Shift-down lemma can be found in Figure~\ref{fig:shift-down}.  
\begin{figure}[H]
\begin{center}
\includegraphics[width=.55\linewidth]{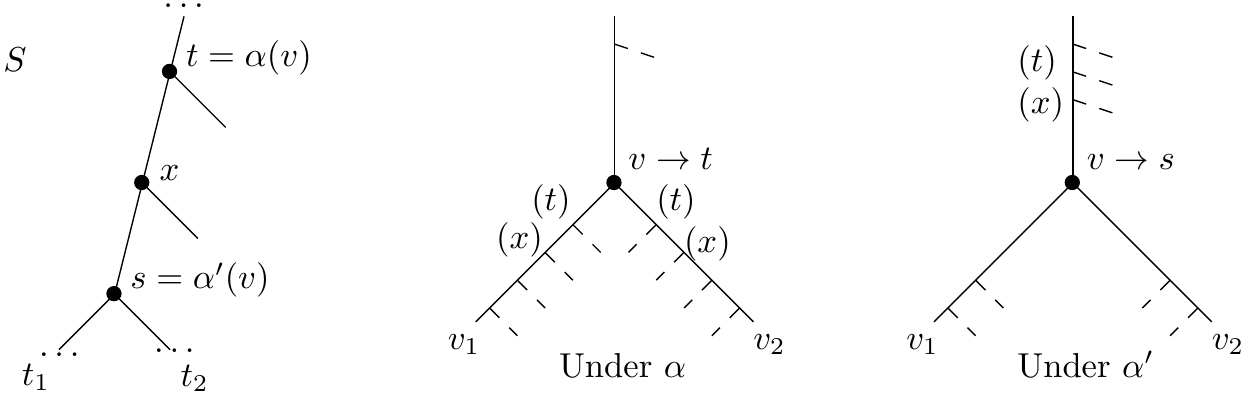}
\caption{The Shift-down lemma in action.  Here $t = par^2(s)$, and remapping $v$ from $t$ to $s$ saves $2$ losses --- $4$ losses are saved below $v$ and $2$ are added above. 
\label{fig:shift-down}}
\end{center}
\end{figure}

\section{The computational complexity of the {\sc MPRST-SD} problem}

We separate the study of the complexity of the {\sc MPRST-SD} problem into two subcases: when 
$\lambda \geq \delta$ and when $\lambda < \delta$. 

\subsection{The case of $\lambda \geq \delta$.}

The following theorem states that, when $\lambda \geq \delta$, the MPR (ie the LCA-mapping) is a solution to the {\sc MPRST-SD} problem.

\begin{theorem}
Let $\G$ and $S$ be an instance of {\sc MPRST-SD}, and suppose that $\lambda \geq \delta$,  Then the LCA-mapping $\lcamap$ is a reconciliation of minimum cost for $\G$ and $S$. 
Moreover if $\lambda > \delta$, $\lcamap$ is the unique 
reconciliation of minimum cost.
\end{theorem}

\begin{proof}
Let $\map$ be a mapping of $\G$ into $S$ of minimum cost. 
Let $v \in V(\G)$ be a minimal node of $\G$ with the property that $\map(v) \neq \lcamap(v)$ (i.e. all proper descendants $v'$ of $v$ satisfy $\map(v) = \lcamap(v)$).
Note that $v$ must exists since, for every leaf $l \in \L(\G)$, 
we have $\map(l) = \lcamap(l)$.
Because $\map(v) \geq \lcamap(v)$, it follows that $\map(v) > \lcamap(v)$.  Denote $s = \lcamap(v)$ and $t = \map(v)$.  
Then there is some $k \geq 1$ such that $t = par^k(s)$.
Consider the mapping $\map' = \map[v \rightarrow s]$.
This possibly increases the sum of duplications by $1$, so that 
$\dupcost(\map') \leq \dupcost(\map) + 1$.
But by the Shift-down lemma, $l(\map') \leq l(\map) - 1$. 
Thus we have at most one duplication but save at least one loss.

If $\lambda > \delta$, this contradicts the optimality of $\map$, implying that $v$ cannot exist and thus that $\map = \lcamap$.  This proves the uniqueness of $\lcamap$ in this case. 

If $\delta = \lambda$, then $\delta \dupcost(\map') + \lambda l(\map') \leq \delta \dupcost(\map) + \lambda l(\map)$.  By applying the above transformation successively on the minimal nodes $v$ that are not mapped to $\lcamap(v)$, we eventually reach the LCA-mapping $\lcamap$ with an equal or better cost than $\map$.
\end{proof}
\subsection{The case of $\delta > \lambda$.}

We show that, in contrast with the $\lambda \geq \delta$ case, the {\sc MPRST-SD} problem 
is NP-hard when $\delta > \lambda$ and the costs are given as part of the input.
More specifically, we show that the problem is NP-hard when one only wants to minimize the sum of duplication heights, i.e. $\lambda = 0$.  Note that if $\lambda > 0$ but 
is small enough, the effect will be the same and the hardness result also holds --- for instance, putting $\delta = 1$ and say $\lambda < \frac{1}{2|V(\G)||V(S)|}$ ensures that even if a maximum number of losses appears on every branch of $\G$, it does not even amount to the cost of one duplication.
The hardness proof is quite technical, and we refer the interested reader to the Appendix for the details.

We briefly outline the main ideas of the reduction.  The reduction is from the Vertex Cover problem, where we are given a graph $G$ and must find a subset of vertices $V' \subseteq V(G)$ of minimum size such that each edge has at least one endpoint in $V'$.  The species tree $S$ and the forest $\G$ are constructed so that, for each vertex $v_i \in V(G)$, there is a gene tree $A_i$ in $\G$ with a long path of duplications, all of which could either be mapped to a species called $y_i$ or another species $z_i$.  We make it so that mapping to $y_i$ introduces one more duplication than mapping to $z_i$, hence we have to ``pay'' for each $y_i$.
We also have a gene tree $C_h$ in $\G$ for each edge $e_h \in E(G)$, with say $v_i$ and $v_j$ being the endpoints of $e_h$.  In $C_h$, there is a large set of duplications $\mathcal{D}$ under the LCA-mapping $\lcamap$.
We make it so that if we either mapped the duplications in $A_i$ or $A_j$ to $y_i$ or $y_j$, respectively, then we may map all the $\mathcal{D}$ 
nodes to $y_i$ or $y_j$ without adding more duplications.  However if we did not choose $y_i$ nor $y_j$, then it will not be possible to remap the $\mathcal{D}$ nodes without incurring a large duplication cost.  Therefore, the goal becomes to choose a minimum number of $y_i$'s from the $A_i$ trees so that for each edge $e_h = \{v_i, v_j\}$, one of $y_i$ or $y_j$ is chosen for the tree $C_h$.  This establishes the correspondence with the vertex cover instance.

\begin{theorem}\label{thm:nphard}
The {\sc MPRST-SD} problem is NP-hard for $\lambda = 0$ and for given $\delta > \lambda$.
\end{theorem}

The above hardness supposes that $\delta$ and $\lambda$ can be arbitrarily far apart.
This leaves open the question of whether MPRST-SD is NP-hard when $\delta$ and $\lambda$ are fixed constants --- in particular when $\delta = 1 + \epsilon$ and $\lambda = 1$, where $\epsilon < 1$ is some very small constant.
We end this section by showing that the above hardness result persists 
even if only one gene tree is given.
The idea is to reduce from the MPRST-SD show hard just above.  Given a species tree $S$ and a gene forest $\G$, we make $\G$ a single tree by incorporating a large  number of speciations (under $\lcamap$) above the root of each tree of $\G$ (modifying $S$ accordingly), then successively joining the roots of two trees of $\G$ under a common parent until $\G$ has only one tree.  

\begin{theorem}\label{thm:nphard-onetree}
The {\sc MPRST-SD} problem is NP-hard for $\lambda = 0$ and for given $\delta > \lambda$, even if only one gene tree is given as input.
\end{theorem}

\begin{proof}

We reduce from the {\sc MPRST-SD} problem in which multiple trees are given.  We assume that $\delta = 1$ and $\lambda = 0$ and only consider duplications --- we use the same argument as before to justify that the problem is NP-hard for very small $\lambda$.  Let $S$ be the given species tree and $\G$ be the given gene forest.  As we are working with the decision version of {\sc MPRST-SD}, assume we are given an integer $t$ and asked whether $cost^{SD}(\G, S, \map) \leq t$ for some $\map$. 
Denote $n = |\L(\G)|$ and let $G_1, \ldots, G_k$ be the $k > 1$ trees of $\G$.
We construct a corresponding instance of a species tree $S'$ and a single gene tree $T$ as follows (the construction is illustrated in Figure~\ref{fig:nphard-onetree}).
Let $S'$ be a species tree obtained by adding $2(t + k)$ nodes ``above'' the root of $S$.  More precisely, first let $C$ be a caterpillar with $2(t + k)$ internal nodes.  Let $l$ be a deepest leaf of $C$.  Obtain $S'$ by replacing $l$ by the root of $S$.
Then, obtain the gene tree $T$ by taking $k$ copies $C_1, \ldots, C_k$ of $C$, and for each leaf $l'$ of each $C_i$ other than $l$,
put $s(l')$ as the corresponding leaf in $S'$.
Then for each $i \in [k]$, replace the $l$ leaf of $C_i$ by the tree $G_i$ (we keep the leaf mapping $s$ of $G_i$), resulting in a tree we call $T_i$.  Finally, let $T'$ be a caterpillar with $k$ leaves $h_1, \ldots, h_k$, and replace each $h_i$ by the $T_i$ tree.  The resulting tree is $T$.
We show that $cost(\G, S, \map) \leq t$ for some $\map$ if and only if $cost(T, S', \map') \leq t + k - 1$ for some $\map'$.

\begin{figure}[tb]
\begin{center}
\includegraphics[width=.7\linewidth]{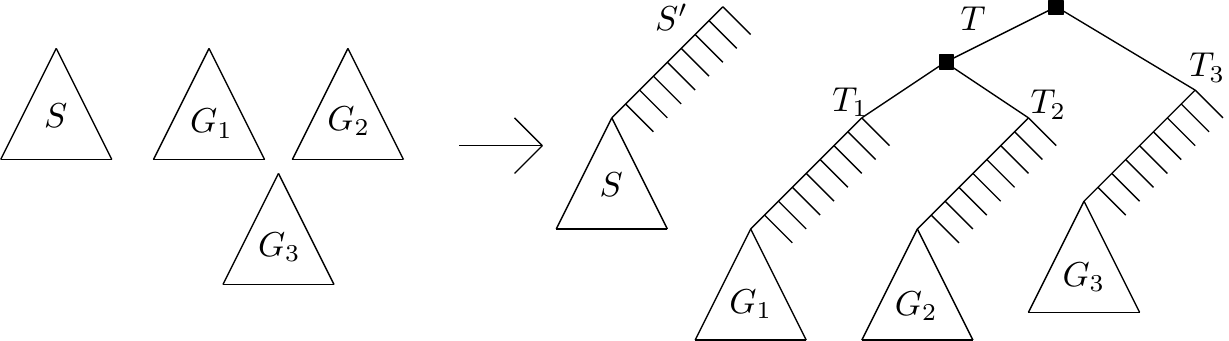}
\caption{The construction of $S'$ and $T$ from $S$ and the set of gene trees $G_1, \ldots, G_k$ (here $k = 3$).  The black squares indicate the path of $k - 1$ duplications that must be mapped to $r(S')$.
\label{fig:nphard-onetree}}
\end{center}
\end{figure}

Notice the following: in any mapping $\map$ of $T$, the $k - 1$ internal nodes of the $T'$ caterpillar must be duplications mapped to $r(S')$, so that $h_{\map}(r(S')) \geq k - 1$.  Also note that under the LCA mapping $\lcamap$ for $T$ and $S'$, the only duplications other than those $k - 1$ mentioned above occur in the $G_i$ subtrees. 
The $(\Rightarrow)$ is then easy to see: given $\map$ such that $cost(\G, S, \map) \leq t$, we set $\map'(v) = \map(v)$ for every node $v$ of $T$ that is also in $\G$ (namely the nodes of $G_1, \ldots, G_k$), and set $\map'(v) = \lcamap(v)$ for every other node.  This achieves a cost of $t + k - 1$.

As for the $(\Leftarrow)$ direction, suppose that $cost^{SD}(T', S', \map') \leq t + k - 1$ for some mapping $\map'$.  Observe that under the LCA-mapping in $T$, each root of each $G_i$ subtree has a path of $2(t + k)$ speciations in its ancestors.  If any node in a $G_i$ subtree of $T$ is mapped to $r(S')$, then all these speciations become duplications (by Lemma~\ref{lem:dups-from-below}), which would contradict $cost^{SD}(T', S', \map') \leq t + k - 1$.  We may thus assume that no node belonging to a $G_i$ subtree is mapped to $r(S')$.
Since $h_{\map'}(r(S')) \geq k - 1$, this implies that the restriction of $\map'$ to the $G_i$ subtrees has cost at most $t$.  

More formally, consider the mapping $\map''$ from $\G$ to $S'$ in which we put $\map''(v) = \map'(v)$ for all $v \in V(\G)$.
Then $cost^{SD}(\G, S', \map'') \leq cost^{SD}(T, S', \map') - (k - 1) \leq t$, because $\map''$ does not contain the top $k - 1$ duplications of $\map'$, and cannot introduce longer duplication paths than in $\map'$.

We are not done, however, since $\map''$ is a mapping from $\G$ to $S'$, and not from $\G$ to $S$.  Consider the set $Q \subseteq V(\G)$ of nodes $v$ of $\G$ such that $\map''(v) \in \overline{V(S)} := V(S') \setminus V(S)$.  We will remap every such node to $r(S)$ and show that this cannot increase the cost.
Observe that if $v \in Q$, then every ancestor of $v$ in $\G$ is also in $Q$. 
Also, every node in $Q$ is a duplication (by invoking Lemma~\ref{lem:remap-dups}).

Consider the mapping $\map^*$ from $\G$ to $S'$ in which we put $\map^*(v) = \map''(v)$ for all $v \notin Q$, and $\map^*(v) = r(S)$ for all $v \in Q$.  It is not difficult to see that $\map^*$ is valid. 



Now, $h_{\map^*}(s) = 0$ for all $s \in \overline{V(S)}$ and $h_{\map^*}(s) = h_{\map''}(s)$ for all $s \in V(S) \setminus \{r(S)\}$.  Moreover, the height of the $r(S)$ duplications under $\map^*$ cannot be more than the height of the forest induced by $Q$ and the duplications mapped to $r(S)$  under $\map''$.  In other words, 

\begin{align*}
h_{\map^*}(r(S)) &\leq \max_{G_i} (h_{\map''}(r(S)) + \sum_{s' \in \overline{V(S)}} h(G_i[\alpha'', s'])) \\
&
=  h_{\map''}(r(S)) + \max_{G_i}( \sum_{s' \in \overline{V(S)}} h(G_i[\alpha'', s'])) \\
&\leq  h_{\map''}(r(S)) + \sum_{s' \in \overline{V(S)}} \max_{G_i}(  h(G_i[\alpha'', s'])) \\
&= h_{\map''}(r(S)) + \sum_{s' \in \overline{V(S)}}  h(\G[\alpha'', s']) \\
&=  h_{\map''}(r(S)) + \sum_{s' \in \overline{V(S)}}  h_{\map''}(s')
\end{align*}

%
%
Therefore, the sum of duplication heights cannot have increased.
Finally, because $\map^*$ is a mapping from $\G$ to $S$, we deduce that $cost^{SD}(\G, S, \map^*) \leq cost^{SD}(\G, S', \map'') \leq t$, as desired.
\end{proof}

\section{An FPT algorithm}

In this section, we show that for costs $\delta > \lambda$ and a parameter $d > 0$, if there is an optimal reconciliation $\map$ of cost $cost^{SD}(\G, S)$ satisfying $\dupcost(\map) \leq d$, 
then $\map$ can be found in time $O(\lceil \frac{\delta}{\lambda}\rceil^d \cdot n \cdot \frac{\delta}{\lambda})$.

In what follows, we allow mappings to be partially defined, and we use the 
$\bot$ symbol to indicate undetermined mappings.  The idea is to start from a mapping in which every internal node is undetermined, and gradually determine those in a bottom-up fashion. 
We need an additional set of definitions.
We will assume that $\delta > \lambda > 0$ (although the algorithm described in this section can solve the $\lambda = 0$ case by setting $\lambda$ to a very small value).

We say that the mapping $\map : V(\G) \rightarrow V(S) \cup \{ \bot \}$ is a \emph{partial mapping}
if $\map(l) = s(l)$ for every leaf $l \in \L(\G)$, 
and it holds that whenever $\map(v) \neq \bot$, we have $\map(v') \neq \bot$ 
for every descendant $v'$ of $v$.  That is, if a node is determined, then all its descendants also are.  This also implies that every ancestor of a $\bot$-node is also a $\bot$-node.
A node $v \in V(\G)$ is a \emph{minimal $\bot$-node} (under $\map$) if $\map(v) = \bot$ and 
$\map(v') \neq \bot$ for each child $v'$ of $v$.
If $\map(v) \neq \bot$ for every $v \in V(\G)$, then $\map$ is called \emph{complete}.
Note that if $\map$ is partial and $\map(v) \neq \bot$, one can already determine whether 
$v$ is an $\S$ or a $\D$ node, and hence we may say that $v$ is a speciation or a duplication under $\map$.  Also note that the definitions of $\dupcost(\map), l(\map)$ and $h_{\map}(s)$ extend naturally to 
a partial mapping $\map$ by considering the forest induced by the nodes not mapped to $\bot$.

If $\map$ is a partial mapping, we call $\map'$ a \emph{completion} of 
$\map$ if $\map'$ is complete, and $\map(v) = \map'(v)$ whenever $\map(v) \neq \bot$.
Note that such a completion always exists, as in particular one can map every $\bot$-node to the root of $S$ (such a mapping must be valid, since all ancestors of a $\bot$-node are also $\bot$-nodes, which ensures that $r(S) = \map'(v) \geq \map'(v')$ for every descendant $v'$ of a newly mapped $\bot$-node $v$). 
We say that $\map'$ is an \emph{optimal completion of $\map$} if $cost^{SD}(\map')$ is minimum among 
every possible completion of $\map$. 
For a minimal $\bot$-node $v$ with children $v_1$ and $v_2$, 
we denote $\lcamap_{\map}(v) = LCA_S(\map(v_1), \map(v_2))$, i.e. the lowest species of $S$ to which 
$v$ can possibly be mapped to in any completion of $\map$.  Observe that $\lcamap_{\map}(v) \geq  \lcamap(v)$.
Moreover, if $v$ is a minimal $\bot$-node, then in any completion $\map'$ of $\map$, $\map'[v \rightarrow \lcamap_{\map}(v)]$ is a valid mapping.
A minimal $\bot$-node $v$ is called a \emph{lowest minimal $\bot$-node} if, for every minimal $\bot$-node $w$ distinct from $v$, 
either $\lcamap_{\map}(v) \leq \lcamap_{\map}(w)$ or $\lcamap_{\map}(v)$ and $\lcamap_{\map}(w)$ are incomparable.

The following Lemma forms the basis of our FPT algorithm, as it allows us to bound the possible mappings of a minimal $\bot$-node.

\begin{lemma}\label{lem:range-of-dups}
Let $\map$ be a partial mapping and let $v$ be a minimal $\bot$-node.
Then for any optimal completion $\map^*$ of $\map$, $\map^*(v) \leq par^{\lceil \delta/\lambda \rceil}(\lcamap_{\map}(v))$.
\end{lemma}

\begin{proof}
Let $\map^*$ be an optimal completion of $\map$ and let $\map' := \map^*[v \rightarrow \lcamap_{\map}(v)]$.  
Note that $\dupcost(\map') \leq \dupcost(\map^*) + 1$.
Now suppose that $\map^*(v) > par^{\lceil \delta/\lambda \rceil}(\lcamap_{\map}(v))$.
Then by the Shift-down lemma, $l(\map^*) - l(\map') > \lceil \delta/\lambda \rceil \geq \delta/\lambda$.
Thus $cost^{SD}(\map^*) - cost^{SD}(\map') > -\delta + \lambda (\delta/\lambda) = 0$.
This contradicts the optimality of $\map^*$.
\end{proof}


A node $v \in V(\G)$ is a \emph{required duplication} (under $\map$) if, in any completion $\map'$ of $\map$, 
$v$ is a duplication under $\map'$.
We first show that required duplications are easy to find.

\begin{lemma}\label{lem:req-dups}
Let $v$ be a minimal $\bot$-node under $\map$, and let $v_1$ and $v_2$ be its two children.
Then $v$ is a required duplication under $\map$ if and only if $\map(v_1) \geq \lcamap(v)$ or $\map(v_2) \geq \lcamap(v)$.
\end{lemma}

\begin{proof}
Suppose that $\map(v_1) \geq \lcamap(v)$, and let $\map'$ be a completion of $\map$.
If $\map'(v) = \map'(v_1)$, then $v$ is a duplication by definition. 
Otherwise, $\map'(v) > \map'(v_1) = \map(v_1) \geq \lcamap(v)$, and $v$ is a duplication by Lemma~\ref{lem:remap-dups}.  
The case when $\map(v_2) \geq \lcamap(v)$ is identical.

Conversely, suppose that $\map(v_1) < \lcamap(v)$ and $\map(v_2) < \lcamap(v)$.  Then $\map(v_1)$ and $\map(v_2)$ must be incomparable descendants of $\lcamap(v)$ 
(because otherwise if e.g. $\map(v_1) \leq \map(v_2)$, then we would have $\lcamap(v) = LCA_S(\lcamap(v_1), \lcamap(v_2)) \leq LCA_S(\map(v_1), \map(v_2)) = \map(v_2)$, whereas we are assuming that $\map(v_2) < \lcamap(v)$).
Take any completion 
$\map'$ of $\map$ such that $\map'(v) = \lcamap(v)$.
To see that $v$ is a speciation under $\map'$, it remains to argue that 
$\map'(v) = \lcamap(v) = LCA_{S}(\map(v_1), \map(v_2))$.
Since $\lcamap(v)$ is an ancestor of both $\map(v_1)$ and $\map(v_2)$, 
we have $LCA_{S}(\map(v_1), \map(v_2)) \leq \lcamap(v)$. 
We also have $\lcamap(v) = LCA_{S}(\lcamap(v_1), \lcamap(v_2)) \leq LCA_{S}(\map(v_1), \map(v_2))$,
and equality follows.
\end{proof}

Lemma~\ref{lem:spec-compl} and Lemma~\ref{lem:same-height-compl} allow 
us to find minimal $\bot$-nodes of $\G$ that are the easiest to deal with, as their mapping in an optimal completion can be determined with certainty.

\begin{lemma}\label{lem:spec-compl}
Let $v$ be a minimal $\bot$-node under $\map$.  If $v$ is not a required duplication under $\map$, 
then $\map^*(v) = \lcamap_{\map}(v)$ for any optimal completion $\map^*$ of $\map$.
\end{lemma}

\begin{proof}
Let $v_1, v_2$ be the children of $v$, and let $\map^*$ be an optimal completion of $\map$.
Since $v$ is not a required duplication, by Lemma~\ref{lem:req-dups} we have $\map(v_1) < \lcamap(v)$ and $\map(v_2) < \lcamap(v)$ and, as argued in the proof of Lemma~\ref{lem:req-dups}, $\map(v_1)$ and $\map(v_2)$ are incomparable.  We thus have that $\lcamap_{\map}(v) = \lcamap(v)$.  
Then $\map^*[v \rightarrow \lcamap(v)]$ is a valid mapping, and $v$ is a speciation 
under this mapping.  Hence $\dupcost(\map^*[v \rightarrow \lcamap(v)]) \leq \dupcost(\map^*)$.  
Then by the Shift-down lemma, this new mapping has fewer losses, 
and thus attains a lower cost than $\map^*$.
\end{proof}

\begin{lemma}\label{lem:same-height-compl}
Let $v$ be a minimal $\bot$-node under $\map$,
and let $\map_{v} := \map[v \rightarrow \lcamap_{\map}(v)]$.
If $\dupcost(\map) = \dupcost(\map_{v})$, then $\map^*(v) = \lcamap_{\map}(v)$ for any optimal completion $\map^*$ of $\map$.
\end{lemma}

\begin{proof}

Let $\map^*$ be an optimal completion of $\map$.  Denote $s := \lcamap_{\map}(v)$,
and assume that $\map^*(v) > s$ (as otherwise, we are done).  
Let $\map' = \map^*[v \rightarrow s]$. 
We have that $l(\map') < l(\map^*)$ by the Shift-down lemma.
To prove the Lemma, we then show that $\dupcost(\map') \leq \dupcost(\map^*)$.
Suppose otherwise that $\dupcost(\map') > \dupcost(\map^*)$.  As only $v$ changed mapping to $s$ to go from $\map^*$ to $\map'$, this implies that $h_{\map'}(s) > h_{\map^*}(s)$ because of $v$.  Since under $\map^*$, no ancestor of $v$ 
is mapped to $s$, it must be that under $\map'$, $v$ is the root of a subtree $T$ of height $h_{\map'}(s)$ of duplications in $s$.  
Since $T$ contains only descendants of $v$, it must also be that $h_{\map_{v}}(s) = h_{\map'}(s)$ (here $\map_v$ is the mapping defined in the Lemma statement).
As we are assuming that $h_{\map'}(s) > h_{\map^*}(s)$, we get 
$h_{\map_{v}}(s) > h_{\map^*}(s)$.
This is a contradiction, since $h_{\map^*}(s) \geq h_{\map}(s) = h_{\map_{v}}(s)$
(the left inequality because $\map^*$ is a completion of $\map$, and the right equality by the choice of $\map_v$). Then $l(\map') < l(\map^*)$ and $\dupcost(\map') \leq \dupcost(\map^*)$ contradicts the fact that $\map^*$ is optimal.
\end{proof}

We say that a minimal $\bot$-node $v \in V(\G)$ is \emph{easy} (under $\map)$ if $v$ falls into one of the cases 
described by Lemma~\ref{lem:spec-compl} or Lemma~\ref{lem:same-height-compl}.
Formally, $v$ is easy if either $v$ is a speciation mapped to $\lcamap_{\map}(v)$ under any optimal completion of $\map$ 
(Lemma~\ref{lem:spec-compl}), or $\dupcost(\map) = \dupcost(\map[v \rightarrow \lcamap_{\map}(v)])$ (Lemma~\ref{lem:same-height-compl}).
Our strategy will be to ``clean-up'' the easy nodes, meaning that we map them to $\lcamap_{\map}(v)$ as prescribed above, 
and then handle the remaining non-easy nodes by branching over the possibilities.
We say that a partial mapping $\map$ is \emph{clean} if every minimal $\bot$-node $v$ satisfies the two following conditions:

\begin{itemize}
\setlength{\itemindent}{.5in}
\item[{[\bf C1]}]
$v$ is not easy;
\item[{[\bf C2]}]
for all duplication nodes $w$ (under $\alpha$ with $\alpha(w) \neq \bot$), either $\map(w) \leq \lcamap_{\map}(v)$ or $\map(w)$ is incomparable with $\lcamap_{\map}(v)$.
\end{itemize}

Roughly speaking, C2 says that all further duplications that may ``appear'' in a completion of $\map$ will be mapped to nodes ``above''  the current duplications in $\map$.  The purpose of C2 is to allow us to create duplication nodes with mappings from the bottom of $S$ to the top.  
Our goal will be to build our $\alpha$ mapping in a bottom-up fashion in $\G$ whilst maintaining this condition.
The next lemma states that if $\map$ is clean and some \emph{lowest} minimal $\bot$-node $v$ gets mapped to species $s$, 
then $v$ brings with it every minimal $\bot$-node that can be mapped to $s$.

\begin{lemma}\label{lem:clean-mapping}
Suppose that $\map$ is a clean partial mapping, and let $\map^*$ be an optimal completion of $\map$.
Let $v$ be a lowest minimal $\bot$-node under $\map$, and let $s := \map^*(v)$.
Then for every minimal $\bot$-node $w$ 
such that $\lcamap_{\map}(w) \leq s$, 
we have $\map^*(w) = s$.
\end{lemma}

\begin{proof} 

Denote $\map' := \map[v \rightarrow s]$.
Suppose first that $s = \lcamap_{\map}(v)$.
Note that since $\map$ is clean, $v$ is not easy, which implies that $h_{\map'}(s) = h_{\map}(s) + 1$.
Since $v$ is a lowest minimal $\bot$-node, if $w$ is a minimal $\bot$-node such that 
$\lcamap_{\map}(w) \leq s$, 
we must have $\lcamap_{\map}(w) = s$, as otherwise $v$ would not have the `lowest' property.
Moreover, because  $v$ and $w$ are both minimal $\bot$-nodes under the partial mapping $\map$, one cannot be the ancestor of the other and so $v$ and $w$ are incomparable.  This implies that mapping $w$ to $s$ under $\map'$ 
cannot further increase $h_{\map'}(s)$ (because we already increased it by $1$ when mapping $v$ to $s$). Thus $\dupcost(\map') = \dupcost(\map'[w \rightarrow s])$, and $w$ is easy under $\map'$ 
and must be mapped to $s$ by Lemma~\ref{lem:same-height-compl}.  This proves the $\map^*(v) = \lcamap_{\map}(s)$ case.

Now assume that $s > \lcamap_{\map}(v)$, and let $w$ be a minimal $\bot$-node with $\lcamap_{\map}(w) \leq s$.
Let us denote $s' := \map^*(w)$.  If $s' = s$, then we are done.
Suppose  that $s' < s$, noting that $h_{\map^*}(s') > 0$ (because $w$ must be a duplication node, due to $\alpha$ being clean).  
If $s' = \lcamap_{\map}(v)$, then $w$ is also a lowest minimal $\bot$-node.  In this case, using the arguments from the previous paragraph and swapping the roles of $v$ and $w$, one can see that $v$ is easy in 
$\map[w \rightarrow s']$ and must be mapped to $s' < s$, a contradiction.
Thus assume $s' > \lcamap_{\map}(v)$.
Under $\map^*$, for each child $v'$ of $v$, we have $\map^*(v') \leq \lcamap_{\map}(v) < s'$, and 
for each ancestor $v''$ of $v$, we have $\map^*(v'') \geq \map^*(v)
= s > s'$. 
Therefore, by remapping $v$ to $s'$, $v$ is the only duplication mapped to $s'$ among its ancestors and descendants.  In other words, because $h_{\map^*}(s') > 0$, we have $\dupcost(\map^*[v \rightarrow s']) \leq \dupcost(\map^*)$.  Moreover by the Shift-down lemma, 
$l(\map^*[v \rightarrow s']) < l(\map^*)$, which contradicts the optimality of $\map^*$.

The remaining case is $s'  > s$. Note that $h_{\map^*}(s) > 0$ (because $v$ must be a duplication node, due to $\alpha$ being clean).   
Since it holds that $v$ is a minimal $\bot$-node, that $\map$ is clean and that $s > \lcamap_{\map}(v)$, it must 
be the case that $\map$ has no duplication mapped to $s$ (by the second property of cleanness).
In particular, $w$ has no descendant that is a duplication mapped to $s$ under $\map$ (and hence under $\map^*$).
Moreover, as $s' = \map^*(w) > s$, $w$ has no ancestor that is a duplication mapped to $s$.  
Thus $\dupcost(\map^*[w \rightarrow s]) \leq \dupcost(\map^*)$, and the Shift-down lemma contradicts the optimality of $\map^*$.
This concludes the proof.
\end{proof}

\begin{figure}[t]
\begin{center}
\includegraphics[width=.8\linewidth]{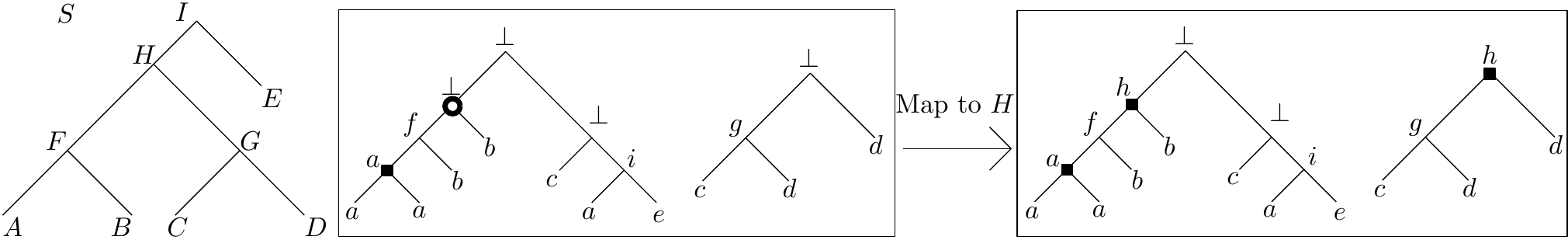}
\caption{An illustration of one pass through the algorithm.  The species tree $S$ is left and $\G$ has two trees (middle) and has partial mapping $\map$ (labels are the lowercase of the species).  Here $\map$ is in a clean state, and the algorithm will pick a lowest minimal $\bot$-node (white circle) and try to map it to, say, $H$.  The forest on the right is the state of $\map$ after applying this and cleaning up.
\label{fig:fpt-algo}}
\end{center}
\end{figure}

We are finally ready to describe our algorithm. 
We start from a partial mapping $\map$ with 
$\map(v) = \bot$ for every internal node $v$ of $\G$.
We gradually ``fill-up'' the $\bot$-nodes of $\map$ in a bottom-up fashion, 
maintaining a clean mapping at each step and ensuring that each decision 
leads to an optimal completion $\map^*$.
To do this, we pick a lowest minimal $\bot$-node $v$, and ``guess'' $\map^*(v)$ 
among the $\lceil \delta/\lambda \rceil$ possibilities.  This increases some $h_{\map}(s)$ by $1$.
For each such guess $s$, we use Lemma~\ref{lem:clean-mapping} to map the appropriate 
minimal $\bot$-nodes to $s$, then take care of the easy nodes to obtain 
another clean mapping.  We repeat until we have either found a complete mapping or we have 
a duplication height higher than $d$.  An illustration of a pass through the algorithm is shown in Figure~\ref{fig:fpt-algo}.

Notice that the algorithm 
assumes that it receives a clean partial mapping $\map$.
In particular, the initial mapping $\map$ that we pass to the first call should satisfy the two properties of cleanness.
To achieve this, we start with a partial mapping $\map$ in which every internal node is a $\bot$-node.
Then, while there is a minimal $\bot$-node $v$ that is not a required duplication, we set $\map(v) = \lcamap_{\map}(v)$, which makes $v$ a speciation.
It is straightforward to see that the resulting $\map$ is clean: C1 is satisfied because we cannot make any more minimal $\bot$-nodes become speciations, and we cannot create any duplication node without increasing $cost^{SD}$ because $\map$ has no duplication.  C2 is met because there are no duplications at all.

\begin{algorithm}
\caption{FPT algorithm for parameter $d$.}\label{alg:fpt-d}
\begin{algorithmic}[1]
\Procedure{SuperReconcile}{$\G, S, \map, d$}
\Statex $\G$ is the set of input trees, $S$ is the species tree, $\map$ is a \emph{clean} partial mapping, $d$ is the maximum value of $\dupcost(\alpha)$.
\If{$\dupcost(\map) > d$} 
  \State Return $\infty$ \label{line:nosol}
\ElsIf{$\map$ is a complete mapping}
  \State Return $cost^{SD}(\map)$   \label{line:returnsol}
\Else  

  \State Let $v$ be a lowest minimal $\bot$-node	\label{line:choose-bot}
  
  \State $bestCost \gets \infty$
  \For{$s$ such that $\lcamap_{\map}(v) \leq s \leq par^{\lceil \delta/\lambda \rceil}(\lcamap_{\map}(v))$} \label{line:try-all-s}

  	\State Let $\map' = \map[v \rightarrow s]$  	\label{line:remap-v}
    
    \For{minimal $\bot$-node $w \neq v$ under $\map$ such that $\lcamap_{\map}(w) \leq s$}	\label{line:remap-other-minimal}
    	\State Set $\map' = \map'[w \rightarrow s]$
    \EndFor
    
  \While{there is a minimal $\bot$-node $w$ that is easy under $\map'$}  \label{line:remap-easy}
  	\State Set $\map' = \map'[w \rightarrow \lcamap_{\map'}(w)]$  
  \EndWhile

  	\State $cost \gets superReconcile(\G, S, \map', d)$	\label{line:reccall}
  	\State \algorithmicif~$cost < bestCost$~\algorithmicthen~$bestCost \gets cost$
  \EndFor

  \State Return $bestCost$
\EndIf
\EndProcedure
\end{algorithmic}
\end{algorithm}


See below for the proof of correctness.  The complexity follows from the fact that the algorithm creates a search tree of degree $\lceil \delta/\lambda \rceil$ of depth at most $d$.  The main technicality is to show that the algorithm maintains a clean mapping before each recursive call\footnote{There is a subtlety to consider here.  
What we have shown is that if there exists a mapping $\map$ of minimum cost $cost^{SD}(\G, S)$ with $\dupcost(\map) \leq d$, then the algorithm finds it.
It might be that a reconciliation $\map$ satisfying $\dupcost(\map) \leq d$ exists, but that the algorithm returns no solution.  This can happen in the case that $\map$ is not of cost $cost^{SD}(\G, S)$.}.

\begin{theorem}\label{thm:algo-works}
Algorithm~\ref{alg:fpt-d} is correct and finds a minimum cost mapping $\map^*$ satisfying $\dupcost(\map^*) \leq d$, if any, in time $O(\lceil \frac{\delta}{\lambda} \rceil^d \cdot n \cdot \frac{\delta}{\lambda})$.
\end{theorem}

\begin{proof}

We show by induction over the depth of the search tree that, in any recursive call made to Algorithm~\ref{alg:fpt-d} with partial mapping $\map$, the algorithm returns the cost of an optimal completion $\map^*$ of $\map$ having 
$\dupcost(\map^*) \leq d$, or $\infty$ if no such completion exists - assuming that the algorithm receives a clean mapping $\map$ as input.
Thus in order to use induction, we must also show that at each recursive call done on line~\ref{line:reccall}, $\map'$ is a clean mapping.  We additionally claim that the search tree created by the algorithm has depth at most $d$.  To show this, we will also prove that every $\map'$ sent to a recursive call satisfies $\dupcost(\map') = \dupcost(\map) + 1$.

The base cases of lines~\ref{line:nosol}~-~\ref{line:returnsol} are trivial.
For the induction step, 
let $v$ be the lowest minimal $\bot$-node chosen on line~\ref{line:choose-bot}.
By Lemma~\ref{lem:range-of-dups}, if $\map^*$ is an optimal completion of $\map$ and $s = \map^*(v)$, then 
$\lcamap_{\map}(v) \leq s \leq par^{\lceil \delta/\lambda \rceil}(\lcamap_{\map}(v))$.  
We try all the $\lceil \delta/\lambda \rceil$ possibilities in the for-loop on line~\ref{line:try-all-s}.
The for-loop on line~\ref{line:remap-other-minimal} is justified by Lemma~\ref{lem:clean-mapping}, 
and the for-loop on line~\ref{line:remap-easy} is justified by Lemma~\ref{lem:spec-compl}
and Lemma~\ref{lem:same-height-compl}.
Assuming that $\map'$ is clean on line~\ref{line:reccall}, by induction the recursive call will return the cost of an optimal completion $\map^*$ of $\map'$
having $\dupcost(\map^*) \leq d$, if any such completion exists.  
It remains to argue that for every $\map'$ sent to a recursive call on line \ref{line:reccall}, $\map'$ is clean and $\dupcost(\map') = \dupcost(\map) + 1$ .

Let us first show that such a $\map'$ is clean, for each choice of $s$ on line~\ref{line:try-all-s}.  
There clearly cannot be an easy node under $\map'$ after 
line~\ref{line:remap-easy}, so we must show C2, i.e. that for any minimal $\bot$-node $w$ under $\map'$, there is no duplication $z$ under $\map'$ satisfying $\map'(z) > \lcamap_{\map'}(w)$.
Suppose instead that $\map'(z) > \lcamap_{\map'}(w)$ for some duplication node $z$.  
Let $w_0$ be a descendant of $w$ that is a minimal $\bot$-node in $\map$ (note that $w_0 = w$ is possible).
We must have $\lcamap_{\map'}(w) \geq \lcamap_{\map}(w_0)$. 
By our assumption, we then have $\map'(z) > \lcamap_{\map'}(w) \geq \lcamap_{\map}(w_0)$.
Then $z$ cannot be a duplication under $\map$, as otherwise $\map$ itself could not be clean (by C2 applied on $z$ and $w_0$).
Thus $z$ is a newly introduced duplication in $\map'$, and so $z$ was a $\bot$-node under $\map$.
Note that Algorithm~\ref{alg:fpt-d} maps $\bot$-nodes of $\G$ one after another, in some order $(z_1, z_2, \ldots, z_k)$.  Suppose without loss of generality that $z$ is 
the first duplication node in this ordering that gets mapped to $\map'(z)$.
There are two cases: either $\map'(z) \neq s$, or $\map'(z) = s$.

Suppose first that $\map'(z) \neq s$.
Lines~\ref{line:remap-v} and~\ref{line:remap-other-minimal} can only map $\bot$-nodes to $s$, 
and line~\ref{line:remap-easy} either maps speciation nodes, 
or easy nodes that become duplications.  Thus when $\map'(z) \neq s$, we may assume that $z$ falls into the latter case, i.e. $z$ is easy before being mapped, so that mapping $z$ to $\map'(z)$ does not increase $h_{\map'}(\map'(z))$. Because $z$ is the first $\bot$-node that gets mapped to $\map'(z)$, this is only possible if there was already a duplication $z_0$ mapped to $\map'(z)$ in $\alpha$.  This implies that $\map(z_0) = \map'(z) > \lcamap_{\map}(w_0)$, and that $\alpha$ was not clean (by C2 applied on $z_0$ and $w_0$).  This is a contradiction.  

We may thus assume that $\map'(z) = s$.  This implies $\lcamap_{\map'}(w) < \map'(z) = s$. 
If $w$ was a minimal $\bot$-node in $\map$, it would have been mapped to $s$ on line~\ref{line:remap-other-minimal}, and so in this case $w$ cannot also be a minimal $\bot$-node in $\map'$, as we supposed.  If instead $w$ was not a minimal $\bot$-node in $\map$, then $w$ has a descendant $w_0$ that was a minimal $\bot$-node under $\map$.  We have $\lcamap_{\map}(w_0) \leq \lcamap_{\map'}(w) < s$, which implies that $w_0$ gets mapped to $s$ on line~\ref{line:remap-other-minimal}.  This makes $\lcamap_{\map'}(w) < s$ impossible, and we have reached a contradiction.
We deduce that $z$ cannot exist, and that $\map'$ is clean.

It remains to show that $\dupcost(\map') = \dupcost(\map) + 1$. 
Again, let $s$ be the chosen species on line~\ref{line:try-all-s}.
Suppose first that $s = \lcamap_{\map}(v)$.  Then $h_{\map[v \rightarrow s]}(s) = h_{\map}(s) + 1$, 
as otherwise $v$ would be easy under $\map$, contradicting its cleanness.
In this situation, as argued in the proof of Lemma~\ref{lem:clean-mapping}, 
each node $w$ that gets mapped to $s$ on line~\ref{line:remap-other-minimal} or on line~\ref{line:remap-easy} is easy, 
and thus cannot further increase the height of the duplications in $s$.  
If $s > \lcamap_{\map}(v)$, then $h_{\map[v \rightarrow s]}(s) = 1 = h_{\map}(s) + 1$, 
since by cleanness no duplication under $\map$ maps to $s$.
Here, each node $w$ that gets mapped on line~\ref{line:remap-other-minimal}
has no descendant nor ancestor mapped to $s$, and thus the height does not increase.
Noting that remapping easy nodes on line~\ref{line:remap-easy}
cannot alter the duplication heights, we get in both cases that $\dupcost(\map[v \rightarrow s]) = \dupcost(\map) + 1$.
This proves the correctness of the algorithm.

As for the complexity, the algorithm creates a search tree of degree $\lceil \delta/\lambda \rceil$ and of depth at most $d$.
Each pass can easily seen to be feasible in time $O(\delta/\lambda \cdot n)$ (with appropriate pre-parsing to compute $\lcamap_{\map}(v)$ in constant time, and to decide if a node is easy or not in constant time as well), and so 
the total complexity is $O(\lceil \delta/\lambda \rceil^d n \cdot \frac{\delta}{\lambda})$.
\end{proof}



\section{Experiments}

\begin{wrapfigure}{R}{.4\linewidth}
\vspace{-.55cm}
\begin{center}
\includegraphics[width=\linewidth]{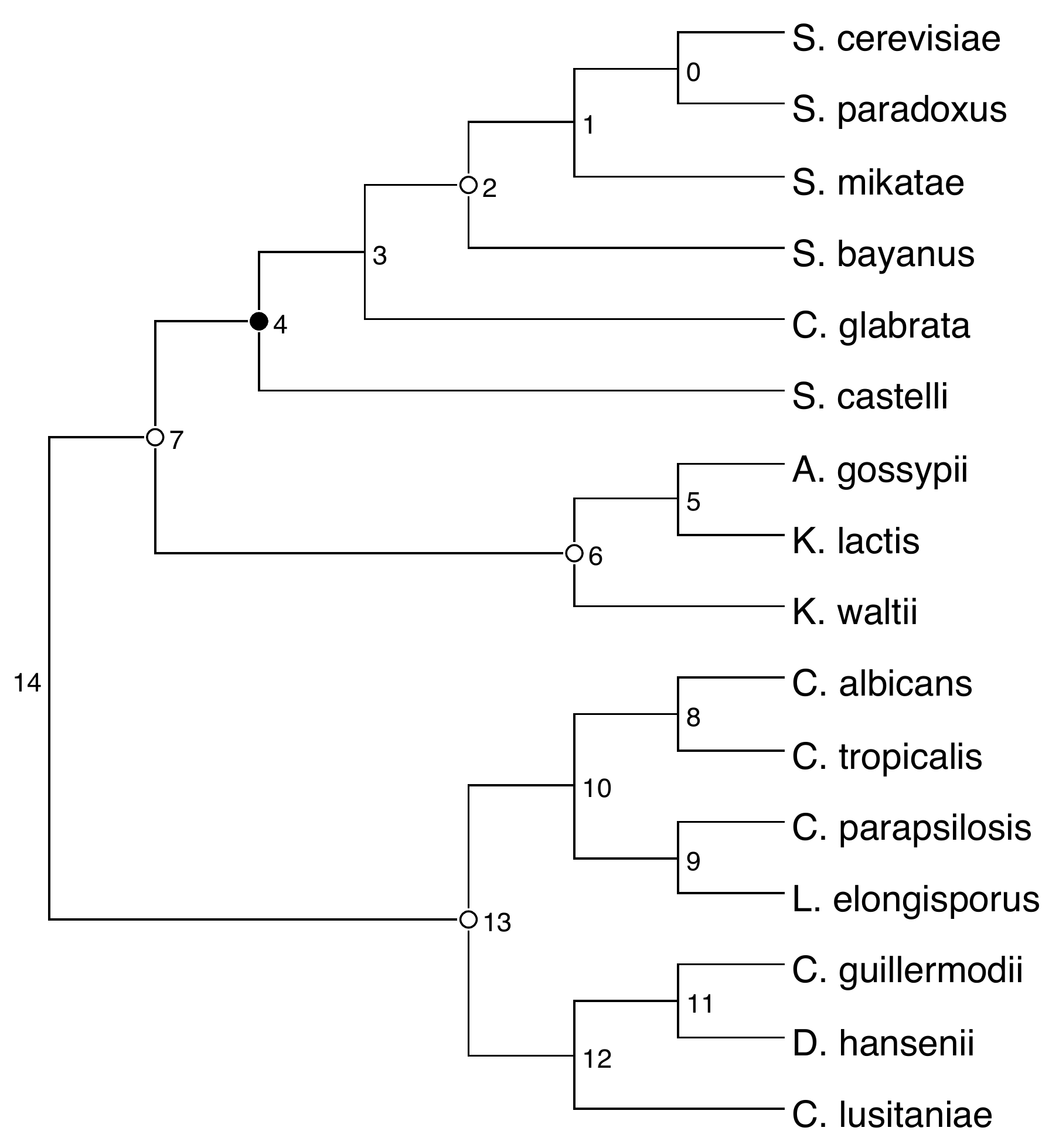}\caption{The species tree phylogeny for the yeast data set described in \cite{butler2009evolution}. Numbers at the internal nodes are meaningless and are only used to refer to the nodes in the main text.}
\label{fig:ST_yeast}
\end{center}
\vspace{-1.6cm}
\end{wrapfigure}

We used our software to reanalyze  a data set of 53 gene trees for 16 eukaryotes presented in \cite{guigo1996reconstruction}
and already reanalyzed in \cite{page2002vertebrate,bansal2008multiple}.
In \cite{bansal2008multiple}, the authors showed that, if segmental duplications are not accounted for, we get a solution having $\dupcost$ equal to 9, while their software (ExactMGD) returns a solution with $\dupcost$ equal to 5. We were able to retrieve the solution with maximum height of 5 fixing $ \delta \in [28, 61]$ and $\lambda=1$, but, as soon as $\delta > 61$, we got a solution with maximum height of 4
where no duplications are placed in the branch leading to the Tetrapoda clade (see \cite[Fig. 1]{page2002vertebrate}) while the other  locations of segmental duplications inferred in \cite{guigo1996reconstruction} are confirmed\footnote{Note that for this data set we used a high value for $\frac{\delta}{\lambda}$ since, because of the sampling strategy, we expect that all relevant genes have been sampled (recall  that in ExactMGD, $\lambda$ is implicitly set to 0).}. This may sow some doubt on the actual existence of a segmental duplication in the LCA of the Tetrapoda clade.

We also reanalyzed the data set of yeast species described in \cite{butler2009evolution}. First, we selected from the data set the 2379 gene trees containing all 16 species and refined unsupported branches using the method described in \cite{DBLP:journals/bioinformatics/JacoxWTS17} and implemented in ecceTERA \cite{ecceTERA} with a bootstrap threshold of 0.9 and $\delta=\lambda=1$. Using our method with $\delta=1.5$, $\lambda=1$ we were able to detect the ancient genome duplication in Saccharomyces cerevisiae already established using synteny information \cite{kellis2004proof}, with 216 gene families supporting the event.
Other nodes with a signature of segmental duplication are nodes 7, 6, 13 and 2 (refer to Fig. \ref{fig:ST_yeast}) with respectively  190, 157, 148 and 136 gene families supporting the event. It would be interesting to see if the synteny information supports these hypotheses.

\section{Conclusion}



This work poses a variety of questions that deserve further investigation.
The complexity of the problem when $\delta/\lambda$ is a constant remains an open problem.  Moreover, our FPT algorithm can handle data sets with a sum of duplication height of about $d = 30$, but in the future, one might consider whether there exist fast approximation algorithms for MPRST-SD in order to attain better scalability.
Other future directions include a multivariate complexity analysis of the problem, in order to understand whether
it is possible to identify other parameters that are small in practice.
Finally, we plan to extend the experimental analysis to other data sets, for instance for the detection of whole genome duplications in plants.
 
 \bibliography{biblio}


\section*{Appendix}

\subsection*{Proof of Theorem~\ref{thm:nphard}}

In this section, we show the hardness of MPRST-SD with non-fixed $\delta$ and $\lambda$.

We will first need particular trees as described by the following Lemma.
These trees guarantee that a prescribed set of leaves $L$ are at distance exactly $k$ from the root, and any two of the leaves in $L$ have their LCA at distance at least $k/2$.
Recall that for a tree $T$ and $u,v \in V(T)$, $dist_T(u,v)$ denotes the number of edges 
on the path between $u$ and $v$ in $T$ (we write $dist(u,v)$ for short).
A \emph{caterpillar} is a binary rooted tree in which each internal node has one child that is a leaf, 
with the exception of one internal node which has two such children.

\begin{lemma} \label{lem:well-behaved-trees}
Let $k \geq 8$ be an integer, and let 
$L$ be a given set of at most $k$ labels.
Then there exists a rooted tree $T$ with leaf set $L'$ with $L \subseteq L'$ 
such that for any $l \in L$, $dist(l, r(T)) = k$ and 
for any two distinct $l_1, l_2 \in L$, $dist(l_1,  LCA_{T}(l_1, l_2)) \geq k/2$.
Moreover, $T$ can be constructed in polynomial time with respect to $k$.
\end{lemma}

\begin{proof}
Let $p$ be the smallest integer such that $k \leq 2^p$.
First consider a fully balanced binary tree $T$ on $2^{p}$ leaves, so that each leaf is at distance $p$ from the root.
Then replace each leaf by the root of a caterpillar of height $k - p + 1$ (hence in the caterpillars the longest root-to-leaf path has $k - p$ edges). The resulting tree is $T$.
Choose $k$ of these caterpillars, and in each of them, assign 
a distinct label of $L$ to a deepest leaf (any of the two).
Thus $dist(l, r(T)) = k$ for each $l \in L$, and clearly $T$ can be built in polynomial time.
We also have $dist(l_1, LCA_{S}(l_1, l_2)) \geq
k - p$ for each distinct $l_1, l_2 \in L$.
As $2^p \leq 2k$, we have $p \leq \log(2k)$. 
This implies $k - p \geq k - \log(2k) \geq k/2$ for $k \geq 8$.
\end{proof}
In the following, we will assume that $\lambda = 0$ and $\delta = 1$.
We reduce the Vertex Cover problem to that of finding a mapping of minimum cost for given $\G$ and $S$.
Recall that in Vertex Cover, we are given a graph $G = (V, E)$ and an integer $\beta < n$ and are asked if there exists a subset $V' \subseteq V$ with $|V'| \leq \beta$ such that every edge of $E$ has at least one endpoint in $V'$.  
For such a given instance, denote $V = \{v_1, \ldots, v_n\}$ and $E = \{e_1, \ldots, e_m\}$ (so that $n = |V|$ and $m = |E|$).
The ordering of the  $v_i$'s and $e_j$'s can be arbitrary, but must remain fixed for the remainder of the construction.

Let $K := (n + m)^{10}$, and observe in particular that $\beta < n \ll K$.  We construct a species tree $S$ and a gene forest $\G$ from $G$.
The construction is relatively technical, but we will provide the main intuitions after having fully described it.
For convenience, we will describe $\G$ as a set of gene trees instead of a single graph.
Figure~\ref{fig:nphard-alltrees} illustrates the constructed species tree and gene trees.
The construction of $S$ is as follows:
start with $S$ being a caterpillar on $3n + 2$ leaves.
Let 
$(x_1, y_1, z_1, x_2, y_2, z_2, \ldots, x_n, y_n, z_n, x_{n + 1})$
be the path of this caterpillar  
consisting of the internal nodes, ordered by decreasing depth
(i.e. $x_1$ is the deepest internal node, and $x_{n + 1}$ is the root).
For each $i \in [n]$, call $p_i, q_i$ and $r_i$, respectively,  the leaf child of $x_i, y_i$ and $z_i$. Note that $x_1$ has two leaf children: choose one to name $p_1$ arbitrarily among the two.  
Then for each edge $uv$ of $S$, 
graft a large number, say $K^{100}$, 
of leaves on the $uv$ edge (grafting $t$ leaves on an edge $uv$ consists of subdividing $uv$ $t$ times, thereby creating $t$ new internal nodes of degree $2$, then adding a leaf child to each of these nodes, see Figure~\ref{fig:nphard-alltrees}).  We will refer to these
grafted leaves as the \emph{special $uv$ leaves}, 
and the parents of these leaves as the \emph{special $uv$ nodes}.
These special nodes are the fundamental tool that lets us control the range of duplication mappings.

Finally, for each $i \in \{1, \ldots, n\}$, replace the leaf $p_i$ by a tree $P_i$ that contains $K$ distinguished leaves 
$c_{i,1}, \ldots, c_{i,K}$ such that $dist(c_{i,j}, r(P_i)) = K$ for all $j \in [K]$, and such that $dist(c_{i,j}, lca(c_{i,j}, c_{i,k})) \geq K/2$ for all distinct $j,k \in [K]$.  By Lemma~\ref{lem:well-behaved-trees}, each $P_i$ can be constructed in polynomial time.
Note that the edges inside a $P_i$ subtree do not have special leaves grafted onto them.
This concludes the construction of $S$.

\begin{figure}[h]
\begin{center}
\includegraphics[width=\linewidth]{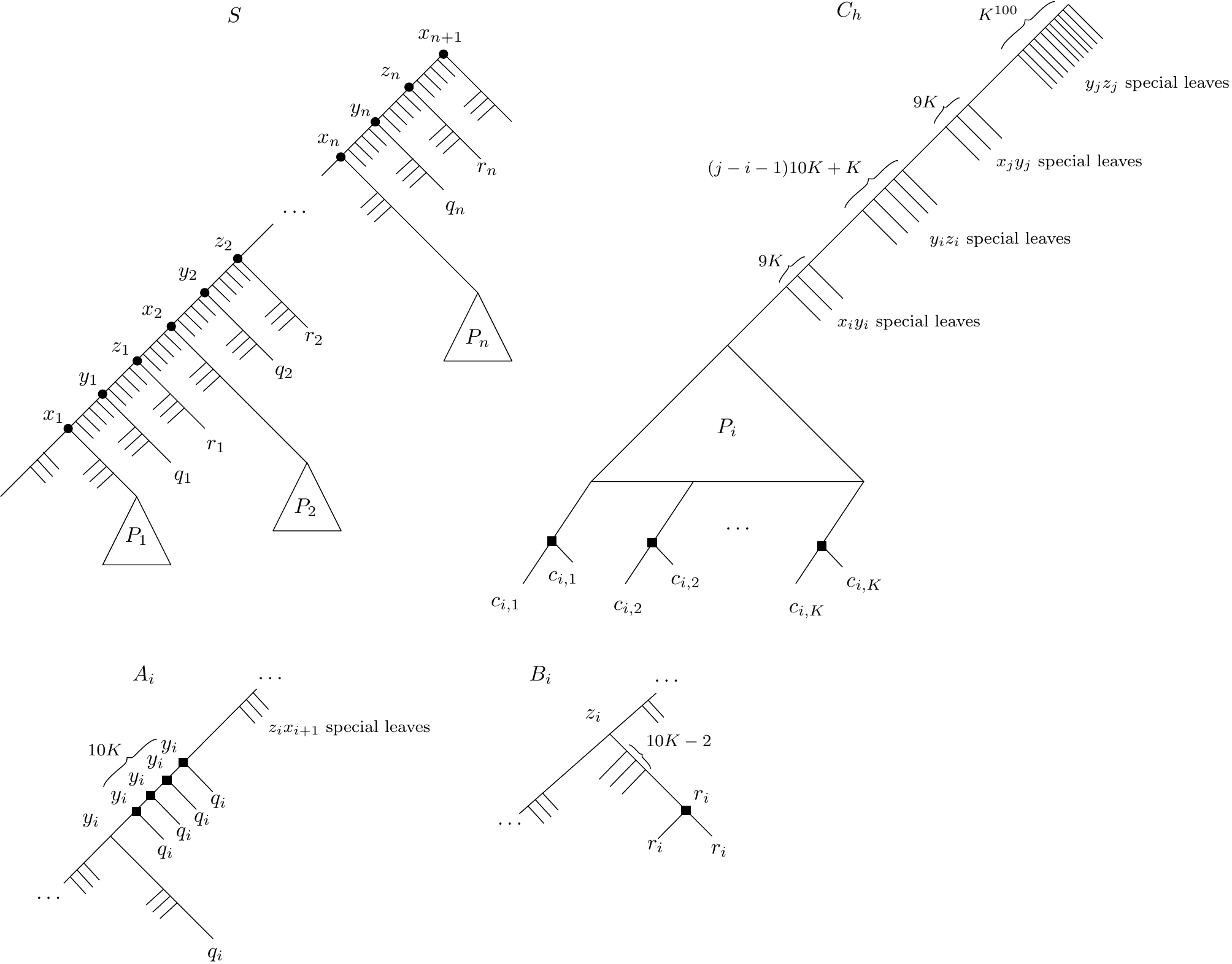}
\caption{The species tree $S$, and the $A_i, B_i$ and $C_h$ trees.  The internal nodes are labeled by their LCA-mapping $\lcamap$, and black squares on the gene trees represent duplication nodes under $\lcamap$.
\label{fig:nphard-alltrees}}
\end{center}
\end{figure}

We proceed with the construction of the set of gene trees $\G$. 
Most of the trees of $\G$ consist of a subset of the nodes of $S$, to which we graft additional leaves to introduce duplications --- some terminology is needed before proceeding.
For $w \in V(S)$, \emph{deleting} $w$ consists in removing $w$ and all its descendants from $S$, 
then contracting the possible resulting degree two vertex (which was the parent of $w$ if $p(w) \neq r(S)$).  If this leaves a root with only one child, we contract the root with its child.
For $X \subseteq \L(S)$, \emph{keeping} $X$ consists of 
successively deleting every node that has no descendant in $X$ until none remains (the tree obtained by keeping $X$ is sometimes called the \emph{restriction} of $S$ to $X$).

The forest $\G$ is the union of three sets of trees $A, B, C$, so that $\G = A \cup B \cup C$.
Roughly speaking $A$ is a set of trees corresponding to the choice of vertices in a vertex cover, $B$ 
is a set of trees to ensure that we ``pay'' a cost of one for each vertex in the cover, and
$C$ is a set of trees corresponding to edges.
For simplicity, we shall describe the trees of $\G$ as having leaves labeled by elements of $\L(S)$ - a leaf labeled $s \in \L(S)$ in a gene tree $T \in \G$ is understood to be a unique gene that belongs to species $s$.

\begin{itemize}
\item
\textbf{The $A$ trees.}
Let $A = \{A_1, \ldots, A_n\}$, one tree for each vertex of $G$.
For each $i \in [n]$, obtain $A_i$ by first taking a copy of $S$, 
then deleting all the special $y_iz_i$ leaves.
Then on the resulting $z_iy_i$ branch, graft $10K$  leaves labeled $q_i$.  Then delete the child of $z_i$ that is also an ancestor of $r_i$ (removing $z_i$ in the process).  Figure~\ref{fig:nphard-alltrees} bottom-left might be helpful.

As a result, under the LCA-mapping $\lcamap$, $A_i$ has a path of $10K$ duplications mapped to $y_i$.  One can choose whether to keep this mapping in $A_i$, or to remap these duplications to $z_i$.

\item
\textbf{The $B$ trees.}
Let $B = \{B_1, \ldots, B_n\}$.
For $i \in [n]$, $B_i$ is obtained from $S$ by deleting all except $10K - 2$ of the special $r_iz_i$ leaves, and grafting a leaf labeled $r_i$ on the edge between $r_i$ and its parent, thereby creating a single duplication mapped to $r_i$ under $\lcamap$.

\item
\textbf{The $C$ trees.}
Let $C = \{C_1, C_2, \ldots, C_m\}$, where for $h \in [m]$, $C_h$ 
corresponds to edge $e_h$.  
Let $v_i, v_j$ be the two endpoints of edge $e_h = \{v_i,v_j\}$, where $i < j$.
To describe $C_h$, we list the set of leaves that we keep from $S$.
Keep all the leaves of the $P_i$ subtree of $S$, and keep a subset of the special leaves defined as follows:  

- keep $9K$ of the special $x_iy_i$ leaves;

- keep $(j - i - 1)10K + K$ of the special $y_iz_i$ leaves;

- keep $9K$ of the special $x_jy_j$ leaves;

- keep all the special $y_jz_j$ leaves.

No other leaves are kept.  Next, in the tree obtained by keeping the aforementioned list of leaves, 
for each $k \in [K]$ we graft, on the edge between $c_{h,k}$ and its parent, another leaf labeled $c_{h,k}$. 
Thus $C_h$ has $K$ duplications, all located at the bottom of the $P_i$ subtree.
This concludes our construction.

\end{itemize}

Let us pause for a moment and provide a bit of intuition for this construction.
We will show that $G$ has a vertex cover of size $\beta$ if and only if there exists a mapping $\map$ of $\G$ of cost at most $10Kn + \beta$.
As we will show later on, two $A_i$ trees cannot have a duplication mapped to the same species of $S$, so these trees alone account for $10Kn$ duplications.  The $r_i$ duplications in the $B_i$ trees account for $n$ more duplications, so that if we kept the LCA-mapping, we would have $10Kn + n > 10Kn + \beta$ duplications.  But these $r_i$ duplications can be remapped to $z_i$, at the cost of creating a path of $10K$ duplications to $z_i$.  This is fine if $A_i$ also has a path of $10K$ duplications to $z_i$, as this does not incur additional height.
Otherwise, this path in $A_i$ is mapped to $y_i$, in which case we leave $r_i$ untouched, summing up to $10K + 1$ duplications for such a particular $A_i, B_i$ pair.  Mapping the duplications of $A_i$ to $y_i$ represents including vertex $v_i$ in the vertex cover, and mapping them to $z_i$ represents not including $v_i$.  Because each time we map the $A_i$ duplications to $y_i$, we have the additional $r_i$ duplication in $B_i$, we cannot do that more than $\beta$ times.

Now consider a $C_h$ tree.  Under the LCA-mapping, the $c_{h,k}$ duplications at the bottom enforce an additional $K$ duplications. 
This can be avoided by, say, mapping all these duplications to the same species.  For instance, we could remap all these duplications to some $y_i$ node of $S$.  But in this case, because of Lemma~\ref{lem:dups-from-below}, every node $v$ of $C_h$ above a $c_{h,k}$ duplication for which $\lcamap(v) \leq y_i$ will become a duplication.  
This will create a duplication subtree $D$ in $C_h$ with a large height, and our goal will be to ``spread'' the duplications of $D$ among the $y_i$ and $z_i$ duplications that are present in the $A_i$ trees.  As it turns out, this will be feasible only if some $y_i$ or $y_j$ has duplication height $10K$, where $v_i$ and $v_j$ are the endpoints of edge $e_h$ corresponding to $C_h$.  If this does not occur, any attempt at mapping the $c_{h,k}$ nodes to a common species will induce a chain reaction of too many duplications created above.
We now proceed with the details.

\begin{theorem}
The {\sc MPRST-SD} problem is NP-hard when $\lambda = 0$.
\end{theorem}

\begin{proof}
Let $G$ and $\beta$ be a given instance of vertex cover, and let $\G$ and $S$ be constructed as  above.
Call a node $v \in V(\G)$ an \emph{original duplication} if $v$ is a duplication under $\lcamap$.  If $\lcamap(v) = s$, we might call $v$ an original $s$-duplication for more precision.  For $T \in \G$ and $t \in V(S)$, suppose there is a unique node $w \in V(T)$ such that $\lcamap(w) = s$.  
We then denote $w$ by $T[t]$.
In particular, any special node that is present in a tree $T \in \G$ satisfies the property, so when we mention the special $uv$ nodes of $T$,
we refer to the special nodes that are mapped to the corresponding special $uv$ nodes in $S$ under $\lcamap$.
For example in the $C_h$ tree of Figure~\ref{fig:nphard-alltrees}, the indicated set of $(j - i - 1)10K + K$ nodes are called special $y_iz_i$ nodes as they are mapped to the special $y_iz_i$ nodes of $S$ under $\lcamap$.

We now show that $G$ has a vertex cover of size $\beta$ if and only if $\G$ and $S$ admit a mapping $\map$ of cost at most $10Kn + \beta$.

($\Rightarrow$) Suppose that $V' = \{v_{a_1}, \ldots, v_{a_{\beta}} \}$ is a vertex cover of $G$.
We describe a mapping $\map$ such that for each $i \in [n]$:

\begin{itemize}
\item
if $v_i \in V'$, 
then $h_{\map}(y_i) = 10K, h_{\map}(r_i) = 1$ and $h_{\map}(z_i) = 0$;
\item
if $v_i \notin V'$, 
then $h_{\map}(z_i) = 10K, h_{\map}(r_i) = 0$ and $h_{\map}(y_i) = 0$ 
\end{itemize}

and $h_{\map}(w) = 0$ for every other node $w \in V(S)$
(we have made explicit the cases in which $h_{\map}(y_i) = 0, h_{\map}(z_i) = 0$
and $h_{\map}(r_i) = 0$ to emphasize them).  

Summing over all $i \in [n]$, a straightforward verification show that this mapping $\map$ attains a cost of 

$$\sum_{s \in V(S)}h_{\map}(s) = \sum_{v_i \in V'}(10K + 1) + \sum_{v_i \notin V'}10K = 10Kn + \beta$$

It remains to argue that each tree can be reconciled using these duplications heights.  
In the remainder, we shall view these duplication heights as ``free to use'', meaning that we are allowed to create a duplication path of nodes mapped to $s \in V(S)$, as long as this path has at most  $h_{\map}(s)$ nodes, using $h_{\map}$ defined above.

Let $A_i \in A$ be one of the $A$ trees, $i \in [n]$.
If $v_i \in V'$ is in the vertex cover, then we have put $h_{\map}(y_i) = 10K$.  
In this case, setting $\map(w) = \lcamap(w)$ for every node $w$ in $A_i$ is a valid mapping in which $h_{\map}(s)$ described above is respected for all $s \in V(S)$ (the only duplications in $A_i$ are those $10K$ mapped to $y_i$).  
If instead $v_i \notin V'$, then we may set $\map(w) = z_i$ for all the $10K$ original $y_i$-duplications of $A_i$, and set $\map(w) = \lcamap(w)$ for every other node.  This is easily seen to be valid since the ancestors of the original $y_i$-duplications in $A_i$ are all proper ancestors of $z_i$.

Let $B_i \in B$ be a $B_i$ tree, $i \in [n]$.  
If $v_i \in V'$, then $h_{\map}(r_i) = 1$, and so setting $\map(w) = \lcamap(w)$ for all $w \in V(B_i)$ is valid and respects $h_{\map}(s)$ for all $s \in V(S)$ (since the $r_i$ duplication is mapped to $r_i$ and there are no other duplications).
If $v_i \notin V'$, then set $\map(w) = z_i$ for every node between the original $r_i$-duplication in $B_i$ and $B_i[z_i]$ and set $\map(w) = \lcamap(w)$ for every other node $w$.  This creates a path of $10K$ duplications mapped to $z_i$, which is acceptable since $h_{\map}(z_i) = 10K$.  This case is illustrated in Figure~\ref{fig:nphard-reconcile}.

\begin{figure}[t]
\begin{center}
\includegraphics[width=.9\linewidth]{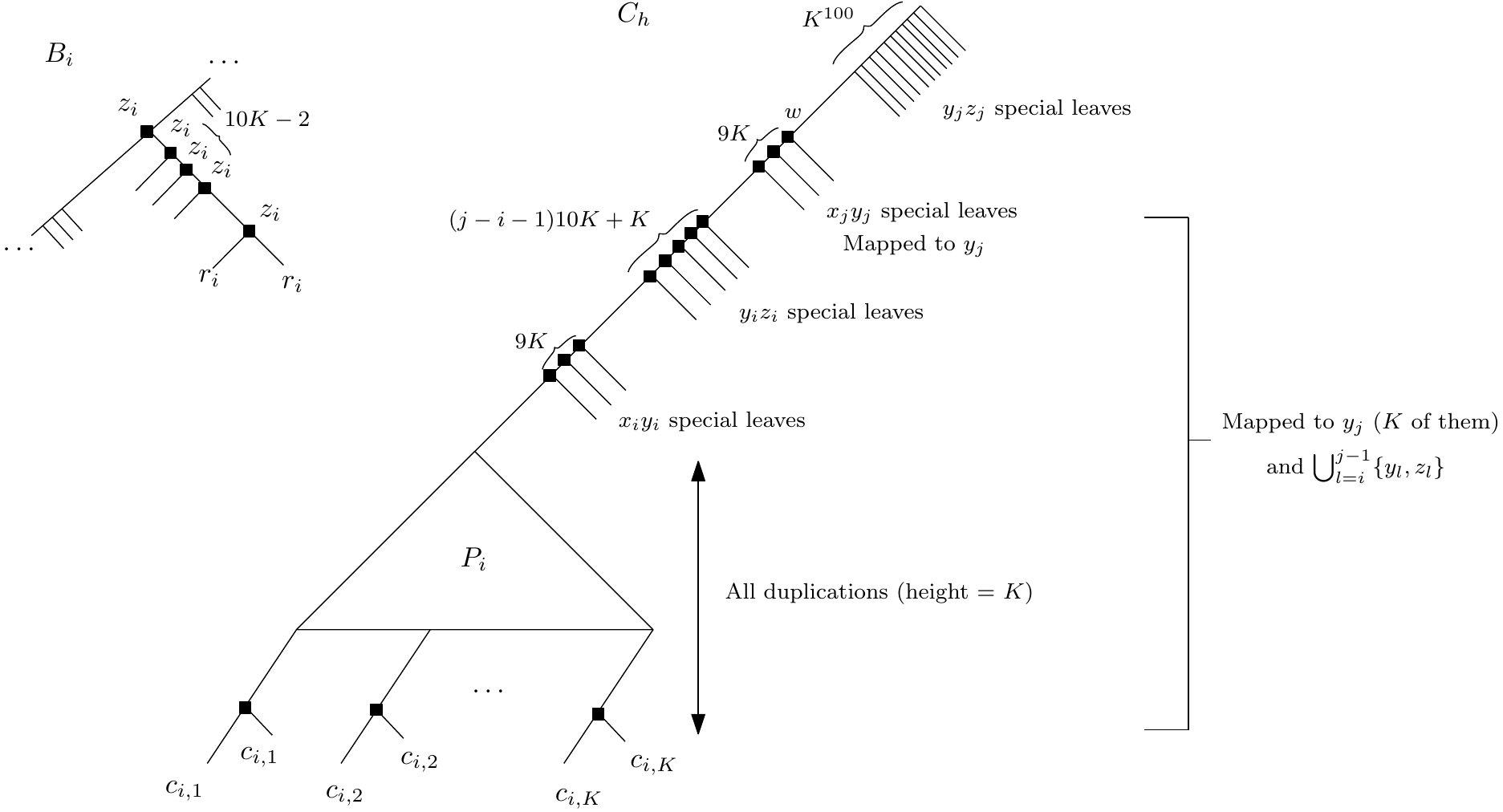}
\caption{Mapping for the $B_i$ tree when $v_i \notin V'$, and the $C_h$ tree when $v_i \notin V'$ but $v_j \in V'$. 
\label{fig:nphard-reconcile}}
\end{center}
\end{figure}

Let $C_h \in W$ be a $C$ tree, $h \in [m]$.
Let $v_i, v_j$ be the two endpoints of edge $e_m = \{v_i, v_j\}$, with $i < j$.
Since $V'$ is a vertex cover, we know that one of $v_i$ or $v_j$ is in $V'$.  Suppose first that $v_i \in V'$.
In this case, we have set $h_{\map}(y_i) = 10K$.
Let $w$ be the highest special $x_{i}y_{i}$ node of $C_h$ (i.e. the closest to the root).  
We set $\map(w) = y_i$ for each internal node descending from $w$.
All of these nodes become duplications, but 
the number of nodes of a longest path from $w$ to an internal node descending from $w$ is $10K$ ($K$ for the $R_i$ subtree, plus $9K$ for the special $x_iy_i$ nodes).  Thus having $h_{\map}(y_i)$ is sufficient to cover the whole subtree rooted at $w$ with duplications.
All the proper ancestors of $w$ can retain the LCA mapping and be speciations, since $\lcamap(w') > y_i$ for all these proper ancestors $w'$.

Now, let us suppose that $v_i \notin V'$, implying $v_j \in V'$.
Note that $h_{\map}(y_l) + h_{\map}(z_l) = 10K$ for all $l \in [n]$.
This time, let $w$ be highest special $x_jy_j$ node of $C_h$.
Again, we will make $w$ and all its internal node descendants duplications.  We map them to the set 
$\{y_j\} \cup \bigcup_{l = i}^{j-1}\{y_l,z_l\}$.  
This case is illustrated in Figure~\ref{fig:nphard-reconcile}.  

More specifically, the longest path from $w$ to a descending internal node contains 
$9K + (j - i - 1)10K + K + 9K + K = (j - i + 1)10K$, where we have counted the special $x_jy_j$ nodes, the special $y_iz_i$ nodes, the special $x_iy_i$ nodes and the $P_i$ subtree nodes.
We have $h_{\map}(y_j) + \sum_{l = i}^{j - 1}(h_{\map}(y_l) + h_{\map}(z_l)) = 10K + (j - i)10K = (j - i + 1)10K$, just enough to map the whole subtree rooted at $w$ to duplications.
It is easy to see that such a mapping can be made valid by first mapping the $9K$ special $x_jy_j$ nodes to $y_j$, then the other nodes descending from $w$ to the rest of $\{y_j\} \cup \bigcup_{l = i}^{j-1}\{y_l,z_l\}$.  This is because all these nodes are ancestors of the special $y_iz_i$ nodes, the special $x_iy_i$ nodes and the $P_i$ nodes (except $y_i$, but we have $h_{\map}(y_i) = 0$ anyway).

We have constructed a mapping $\map$ with the desired duplication heights, concluding this direction of the proof.  Let us proceed with the converse direction.

$(\Leftarrow)$: suppose that there exists a mapping $\map$ of the $\G$ trees of cost at most $10Kn + \beta$.  
We show that there exists a vertex cover of size at most $\beta$ in $G$.  
For some $X \subseteq V(S)$, define $h_{\map}(X) = \sum_{x \in X} h_{\map}(x)$.
For each $i \in [n]$, define the sets 

$$Y_i = \{ s \in V(S) : y_i \leq s < z_i\}$$

and

$$Z_i = \{ s \in V(S) : z_i \leq s < x_{i + 1}\} \quad R_i = \{s \in V(S) : r_i \leq s < z_i\}$$

Our goal is to show that the $Y_i$'s for which $h_{\map}(Y_i) > 9K$ correspond to a vertex cover.
The proof is divided into a series of claims.

\begin{nclaim} \label{claim:basic-bound}
For each $i \in [n]$, 
$h_{\map}(Y_i) + h_{\map}(Z_i) \geq 10K$.
\end{nclaim}

\begin{proof}
Consider the $A_i$ tree of $A$, and let $y_i^*$ be an original $y_i$-duplication in $A_i$.
If $\map(y_i^*) \geq x_{i+1}$, then every node of $A_i$ on the path from $y_i^*$ to $A_i[x_{i+1}]$ is a duplication.  This includes all the $K^{100}$ special   
$z_ix_{i+1}$ nodes, contradicting the cost of $\map$.  Thus $y_i \leq \map(y_i^*) < x_{i + 1}$.  That is, $\map(y_i) \in Y_i \cup Z_i$.  As this is true for all the original $y_i$-duplications of $A_i$, and because $Y_i$ and $Z_i$ are disjoint sets, this shows that $h_{\map}(Y_i) + h_{\map}(Z_i) \geq 10K$.
\end{proof}

The above claim shows that we already need a duplication height of $10nK$ just for the union of the $Y_i$ and $Z_i$ nodes.
This implies the following.

\begin{nclaim} \label{claim:another-basic-bound}
There are at most $\beta$ duplication in $\G$ that are not mapped to a node in $\bigcup_{i \in [n]}(Y_i \cup Z_i)$.
Moreover, for any subset $I \subseteq [n]$, $\sum_{i \in I}(h_{\map}(Y_i) + h_{\map}(Z_i)) \leq 10K|I| + \beta$.
\end{nclaim}

\begin{proof}
The first statement follows from Claim~\ref{claim:basic-bound} and the cost of $\map$.
As for the second statement, 
suppose it does not hold for some $I \subseteq [n]$.
Then $\sum_{i \in [n]} (h_{\map}(Y_i) + h_{\map}(Z_i)) = \sum_{i \in I} (h_{\map}(Y_i) + h_{\map}(Z_i)) + \sum_{i \in [n] \setminus I}(h_{\map}(Y_i) + h_{\map}(Z_i)) > 10Kn + \beta$, a contradiction to the cost of $\map$.
\end{proof}

Now, let $c$ be an original duplication in some tree of $T \in \G$.
The \emph{$c$-duplication} path $\P(c)$ is the maximal path of $T$ with its arcs reversed that 
starts at $c$ and contains only ancestors of $c$ that are duplication nodes under $\map$ (in other words, we start at $c$ and include it in $\P(c)$, traverse the ancestors one after another and include every duplication node encountered, and then stop when reaching a speciation or the root --- therefore every node in $\P(c)$ is a duplication).  
We will treat $\P(c)$ as a set of nodes.
We say that $\P(c)$ ends at node $p$ if $p \in \P(c)$ and $p \geq p'$ for every $p' \in \P(c)$.

\begin{nclaim} \label{claim:dup-path}
Let $C_h \in C$ and let $v_i, v_j$ be the two endpoints of edge $e_h$, with $i < j$. Then there is an original duplication $c \in V(C_h)$ such that 
$\P(c)$ ends at a node $p$ with $\map(p) \geq y_i$. 
\end{nclaim}

\begin{proof}
Let $c_1, \ldots, c_K$ be the original duplication nodes in the $C_h$ tree, which belong to the $P_i$ subtree of $C_h$.
Assume the claim is false, and that $\map(c_k) < y_i$ for every $k \in [K]$.  First observe that if $\map(c_k) \neq \map(c_{k'})$ for every pair $c_k, c_{k'}$ of original 
duplications in $C_h$, 
then $\sum_{k \in [K]} h_{\map}(\map(c_k)) \geq K$.  
Note that $\map(c_k) \in Y_l \cup Z_l$ is impossible for $l < i$, by the placement of the $c_{h,k}$ leaves in $P_i$ in the species tree $S$.
As we further assume that $\map(c_k) < y_i$ for every $k$, none of the $c_k$ duplications is mapped to a 
member of $\bigcup_{l = i}^n (Y_l \cup Z_l)$ either.  
Therefore, none of the $c_k$ duplications is counted in Claim~\ref{claim:basic-bound}, so this implies a cost of at least $10nK + K > 10nK + \beta$, a contradiction.
So we may assume that $\map(c_k) = \map(c_{k'})$ for some distinct original duplications $c_k, c_{k'}$.
Notice that $\alpha(c_k) = \alpha(c_{k'})$ must be a common ancestor of $\lcamap(c_k)$ and $\lcamap(c_{k'})$.  This implies that every node on the path between $c_k$ and $LCA_{C_h}(c_k, c_{k'})$ is a duplication (by Lemma~\ref{lem:dups-from-below}), which in turn implies $|\P(c_k)| \geq K/2$, by the construction of $P_i$.  By assumption, no duplication of $\P(c_k)$ is mapped to a member of $\bigcup_{l = 1}^n (Y_l \cup Z_l)$, and again due to Claim~\ref{claim:basic-bound}, 
the total cost of $\map$ is at least $10nK + K/2 > 10nK + \beta$, a contradiction.
\end{proof}

We can now show that the sets $Y_i$ for which $h_{\map}(Y_i) > 9K$ correspond to a vertex cover.

\begin{nclaim}\label{claim:chosen-are-vc}
Let $e_h \in E$, and let $v_i, v_j$ be the two endpoints of $e_h$, with $i < j$. 
Then at least one of $h_{\map}(Y_i) > 9K$ or $h_{\map}(Y_j) > 9K$ holds.
\end{nclaim}

\begin{proof}
Assume that $h_{\map}(Y_i) \leq 9K$.  By Claim~\ref{claim:dup-path}, there is an original duplication $c$ in $C_h$ such that $\P(c)$ ends at a node $p$ satisfying $\map(p) \geq y_i$.  This implies that every 
special $x_iy_i$ node in $C_h$ is a duplication (by Lemma~\ref{lem:dups-from-below}).  Thus $\P(c)$ contains all these $9K$ nodes, plus $K$ nodes from the $R_i$ subtree.  Since $h_{\map}(Y_i) \leq 9K$, at least 
$K$ of these nodes are mapped to a node outside $Y_i$.  Call $U$ this set of $K$ nodes that are not mapped in $Y_i$.  By the placement of the $P_i$ subtree, none of the nodes of $U$ is mapped to an element of $Y_l \cup Z_l$ for $l < i$.  
Also by Claim~\ref{claim:another-basic-bound}, at most $\beta$ of the $U$ nodes are mapped to a node 
outside of $W := Z_i \cup \bigcup_{l = i + 1}^n (Y_i \cup Z_i)$, so it follows that at least $K - \beta$ nodes of $U$ are mapped to an element of $W$.  
Because all the elements of $W$ are ancestors of the special $y_iz_i$ nodes, this implies that all the special $y_iz_i$ nodes in $C_h$ are duplications, of which there are $(j - i - 1)10K + K$.  

So far, this makes at least $10K + 10K(j - i - 1) + K = (j - i)10K + K$ duplications in $|\P(c)|$.  
Let us now argue that at least $K - 2\beta$ of these are mapped to an ancestor of $y_j$ (not necessarily proper).
If not, then all these duplications are mapped to $\{X\} \cup \bigcup_{l = i}^{j - 1}(Y_i \cup Z_i)$, where $X$ is some subset of $V(S)$ satisfying $|X| \leq \beta$.  
By Claim~\ref{claim:another-basic-bound}, we know that $h_{\map}(X) + \sum_{l = i}^{j - 1}(h_{\map}(Y_i) + h_{\map}(Z_i)) \leq \beta + (j - i)10K + \beta = (j - i)10K + 2\beta < (j - i)10K + K$. 
In fact, this means that at least $K - 2\beta$ of the duplications considered in $\P(c)$ so far that are unaccounted for, which means that they are mapped to an ancestor of $y_j$.  

Because of this, we now get that the $9K$ special $x_jy_j$ nodes of $C_h$ are duplications, which means that in fact, at least $K - 2\beta + 9K = 10K - 2\beta$ duplications of $\P(c)$ are mapped to an ancestor of $y_j$.  
Notice that $\P(c)$ cannot contain a node $w$ with $\map(w) \geq z_j$. 
Indeed, if this were the case, then all the special $y_jz_j$ nodes of $C_h$ would be duplications, of which there are $K^{100}$. 
Thus all the aforementioned $10K - 2\beta$ duplications are mapped to an ancestor of $y_j$ but a proper descendant of $z_j$, i.e. they are mapped in $Y_j$.  So $h_{\map}(Y_i) \geq 10K - 2\beta > 9K$, proving our claim.
\end{proof}

Let $V' = \{v_i : h_{\map}(v_i) > 9K\}$.  Then Claim~\ref{claim:chosen-are-vc} implies that $V'$ is a vertex cover.  It only remains to show that $|V'| \leq \beta$.  This will follow from our last claim.

\begin{nclaim}\label{claim:chosen-are-costly}
Suppose that $h_{\map}(Y_i) > 9K$.  Then $h_{\map}(Y_i) + h_{\map}(Z_i) + h_{\map}(R_i) \geq 10K + 1$.
\end{nclaim}

\begin{proof}
Let $r_i^*$ be the original $r_i$-duplication in the $A_i$ tree.  
We show that $\map(r_i^*) < z_i$.  If $\map(r_i^*) \geq z_i$, then all the $10K$ nodes on the path from $r_i^*$ to $A_i[z_i]$, including $r_i^*$ itself, are duplications mapped to a node in $Z_i$ (these nodes cannot be mapped to an ancestor of the $Z_i$ nodes, due to the presence of the special $z_ix_{i+1}$ nodes above $A_i[z_i]$). 
Thus $h_{\map}(Z_i) \geq 10K$, and so $h_{\map}(Y_i) + h_{\map}(Z_i) \geq 19K$, contradicting Claim~\ref{claim:another-basic-bound}.
It follows that $\map(r_i^*) \notin Z_i$.  Since we cannot have $\map(r_i^*) \in Y_i$, using Claim~\ref{claim:basic-bound} we have $h_{\map}(Y_i) + h_{\map}(Z_i) + h_{\map}(R_i) \geq 10K + 1$.
\end{proof}

To finish the argument, Claim~\ref{claim:chosen-are-costly} implies that $|V'| \leq \beta$, since each $v_i \in V'$ implies that $h_{\map}(Y_i) > 9K$ and that $h_{\map}(Y_i) + h_{\map}(Z_i) + h_{\map}(R_i) \geq 10K + 1$.
More formally, if we had $|V'| > \beta$, letting $I = \{i : h_{\map}(Y_i) > 9K\}$, with $|I| > \beta$, this would imply 
$\sum_{i \in [n]}(h_{\map}(Y_i) + h_{\map}(Z_i) + h_{\map}(R_i)) \geq \sum_{i \in I}(h_{\map}(Y_i) + h_{\map}(Z_i) + h_{\map}(R_i)) + \sum_{i \in [n] \setminus I}(h_{\map}(Y_i) + h_{\map}(Z_i)) > 10Kn + \beta$.
This concludes the proof.
\end{proof}

\end{document}